\documentclass[11pt,onecolumn]{article}

\usepackage[cmex10]{amsmath}

\usepackage{amsfonts}
\usepackage{amssymb}
\usepackage{cite}
\usepackage[amsmath,thmmarks]{ntheorem}
\usepackage[dvips]{graphicx}
\usepackage{subfigure}

\usepackage{color, soul}
\usepackage{epic, eepic}
\usepackage[top=1in, bottom=1in, left=1.25in, right=1.25in]{geometry}
\usepackage{appendix}

\newcommand{\ud}{\mathrm{d}}

\newcommand{\tr}{\mathrm{tr}}
\newcommand\bs[1]{\boldsymbol{#1}}

\newtheorem{thm}{Theorem}
\newtheorem{prop}{Proposition}
\newtheorem{lem}{Lemma}
\newtheorem{rem}{Remark}
\newtheorem{assump}{Assumption}
\newtheorem{corol}{Corollary}

\theoremheaderfont{\sc}\theorembodyfont{\upshape}
\theoremstyle{nonumberplain}
\theoremseparator{}
\theoremsymbol{\rule{1ex}{1ex}}
\newtheorem{proof}{Proof}

\newcommand{\blue}{\textcolor[rgb]{0.00,0.00,0.00}}

\begin{document}

\title{On the Performance Bound of Sparse Estimation with Sensing Matrix Perturbation}

\author{Yujie~Tang, Laming~Chen, and~Yuantao~Gu\thanks{This work was partially supported by National Natural Science Foundation of
China (NSFC 60872087 and NSFC U0835003). The authors are with Department of Electronic Engineering, Tsinghua University, Beijing 100084,
 China. The corresponding author of this paper is Yuantao Gu (Email: gyt@tsinghua.edu.cn).}}

\date{Received January 3, 2013, Revised April 12, 2013.}

\maketitle

\begin{abstract}
This paper focusses on the sparse estimation in the situation where both the the sensing matrix and the measurement vector are corrupted by additive Gaussian noises. The performance bound of sparse estimation is analyzed and discussed in depth. Two types of lower bounds, the constrained Cram\'{e}r-Rao bound (CCRB) and the Hammersley-Chapman-Robbins bound (HCRB), are discussed. It is shown that the situation with sensing matrix perturbation is more complex than the one with only measurement noise. For the CCRB, its closed-form expression is deduced. It demonstrates a gap between the maximal and nonmaximal support cases. It is also revealed that a gap lies between the CCRB and the MSE of the oracle pseudoinverse estimator, but it approaches zero asymptotically when the problem dimensions tend to infinity. For a tighter bound, the HCRB, despite of the difficulty in obtaining a simple expression for general sensing matrix, a closed-form expression in the unit sensing matrix case is derived for a qualitative study of the performance bound. It is shown that the gap between the maximal and nonmaximal cases is eliminated for the HCRB. Numerical simulations are performed to verify the theoretical results in this paper.

\end{abstract}

\section{Introduction}

The problem of sparse recovery from linear measurement has been a hot topic these years and has drawn a great deal of attention. Various practical algorithms of sparse recovery have been proposed and theoretical results have been derived \cite{Donoho_Stable,Tropp_Just,Candes_Dantzig,Tropp_OMP,Needell_ROMP,Needell_CoSaMP,Dai_SP,Blumensath_IHT,Gu_ZAP,Gu_l1ZAP}. The theory of sparse recovery can be applied to various fields, especially the field of compressive sensing which considerably decreases the sampling rate of sparse signals \cite{Candes_UncertaintyP,Donoho_CS,Candes_CS}.

Suppose that a sparse signal $\mathbf{x}\in\mathbb{R}^n$ is observed through noisy linear measurement
\begin{equation}\label{pre_eq}
\mathbf{y}=\mathbf{A}\mathbf{x}+\mathbf{n},
\end{equation}
where $\mathbf{A}\in\mathbb{R}^{m\times n}$ is called the sensing matrix, $\mathbf{y}\in\mathbb{R}^m$ is the measurement vector, and $\mathbf{n}\in\mathbb{R}^m$ is the additive random noise vector. The main issue of sparse recovery is to estimate $\mathbf{x}$ from measurement $\mathbf{y}$ with estimation error as small as possible, and the recovery algorithm should be computationally tractable.

The performance of various recovery algorithms in noisy scenarios has been theoretical analyzed \cite{Needell_CoSaMP,Dai_SP,Blumensath_IHT,Candes_StableRec,Chartrand_nonconvex,Cai_Stable,Cai_OMPn,Needell_ROMPn}. Most of these works only consider the upper bound of the estimation error. Theoretical result about to what extent the estimation error can be small (i.e. the theoretical lower bound of estimation error) is of great interest because it sets a limit performance which all sparse recovery algorithms cannot exceed. There are various approaches that try to handle this topic. Reference \cite{Candes_HowWell} employed a minimax approach to study the problem. Another approach is to reformulate the sparse recovery problem as a parameter estimation problem \cite{Eldar_SparseCRB}. The sparse vector $\mathbf{x}$ is viewed as a deterministic parameter vector, and $\mathbf{y}$ represents the observation data. The goal of this approach is to minimize the mean-squared error $E[(\hat{\mathbf{x}}-\mathbf{x})^2]$ (MSE) among all possible estimators $\hat{\mathbf{x}}=\hat{\mathbf{x}}(\mathbf{y})$. The theory of lower bounds of MSE has been well established for parameter vector $\mathbf{x}\in\mathbb{R}^n$ without further constraints \cite{Lehmann_PointEstimation}. Various bounds, including the Cram\'{e}r-Rao bound \cite{Kay_StatisticalSP,Lehmann_PointEstimation}, the Hammersley-Chapman-Robbins bound\cite{Gorman_HCRB,CR_HCRB}, and the Barankin bound \cite{BarabkinBound}, have been introduced. However, the classical theory in general requires some modification to adapt to the sparse settings.

Recently, researches on the lower bounds of MSE for constrained parameter vectors, especially sparse parameter vectors, have been developed. The Cram\'{e}r-Rao bound has been modified for the constrained parameter case, and works well in the sparse settings \cite{Eldar_SingFIM,Eldar_SparseCRB}. The Hammersley-Chapman-Robbins bound requires little essential modifications, and has also been applied to the the problem of sparse recovery \cite{Eldar_SSNM,Hormati}.

This paper also focusses on the theoretical lower bounds of sparse estimators and employs the constrained Cram\'{e}r-Rao bound and the Hammersley-Chapman-Robbins bound, but deals with a more general setting in which the sensing matrix is perturbed by additive random noise. Perturbed sensing matrices appear in many practical scenarios, and therefore it is necessary to study the theoretical bounds of sparse recovery with perturbed sensing matrix \cite{Herman_PerturbCS,Zhu_PerturbCS,Ding_PerturbOMP}. One of the consequences of perturbed sensing matrix is that it is a kind of multiplicative noise, and the total noise on the measurement vector is dependent on the parameter vector $\mathbf{x}$, which demonstrates potential complexity compared to the sensing matrix perturbation-free setting \eqref{pre_eq}.

The problem of sensing matrix perturbation is also closely related to other research topics.
In \cite{Chi_BasisMismatch}, the authors considered the case of basis mismatch between the assumed model and the actual one, and commented that sensing matrix perturbation can reflect basis mismatch. They derived the approximation error in terms of mismatch level and showed that extra care may be needed to account for the effects of basis mismatch. In \cite{Eldar_NoiseFolding,Davenport_ProsCons}, the authors analyzed the case where the sparse signal is directly contaminated with white noise and demonstrated the phenomenon of noise folding. Though different from the situation where it is the sensing matrix that is corrupted, their work shows the necessity to study more general noise settings.

The main contributions of this work are the theoretical bounds of sparse recovery with perturbed sensing matrix and noisy measurement vector. Closed-form expressions of the constrained CRB will be derived, and the quantitative behavior will be discussed. For the Hammersley-Chapman-Robbins bound, only the case of identity sensing matrix is studied for the sake of simplicity, but the results are still inspiring in that its analysis is much simpler and can still provide much information about the behavior of the theoretical lower bounds when the noises are large.

The rest of this paper is organized as follows. In Section \ref{sec:Fund_problem}, the fundamental problem of sparse recovery with perturbed sensing matrix is introduced, and the classical theory of parameter estimation will be reviewed. In Section \ref{sec:CCRB}, the constrained Cram\'{e}r-Rao bound will be derived, and quantitative analysis will be provided in order to have a deeper understanding of its behavior. In Section \ref{sec:HCRB}, the Hammersley-Chapman-Robbins bound will be derived for the case with unit sensing matrix, and its behavior with different settings of signals and noises will also be studied. In Section \ref{sec:num_results}, numerical results will be presented to verify the theoretical results \blue{and to compare with existing estimators}. This paper is concluded in Section~\ref{sec:conclusion} and the proofs are postponed to Appendices.

\subsection*{Notation}

The $M\times M$ unit matrix is denoted by $\mathbf{U}_M$.
For any index set $\Lambda\subset\{1,\ldots,N\}$, $|\Lambda|$ denotes the cardinality of $\Lambda$,
and $\Lambda^c$ denotes the complement set $\{1,2,\ldots,N\}\backslash \Lambda$.
For any index set $\Lambda$ and any $N$-dimensional vector $\mathbf{v}$ ($N\geq |\Lambda|$),
$\mathbf{v}_\Lambda$ denotes the $|\Lambda|$-length vector containing
the entries of $\mathbf{v}$ indexed by $\Lambda$. For any index set $\Lambda$ and any $M\times N$ matrix $\mathbf{M}$ ($N\geq|\Lambda|$),
$\mathbf{M}_\Lambda$ denotes the $M\times|\Lambda|$ matrix containing the columns of $\mathbf{M}$ corresponding
to $\Lambda$. For any vector $\mathbf{v}$, $\|\mathbf{v}\|_{\ell_p}$ denotes the $p$-norm of $\mathbf{v}$.
For any appropriate matrix $\mathbf{M}$, $\mathbf{M}^\dagger$ denotes the Moore-Penrose pseudo-inverse of $\mathbf{M}$.
For $\mathbf{x}=(x_1,\ldots,x_N)^\mathrm{T}$, $\nabla_{\mathbf{x}}$ denotes the gradient operator
$(\partial/\partial x_1,\ldots,\partial/\partial x_N)^\mathrm{T}$, and $\nabla^\mathrm{T}_{\mathbf{x}}$ denotes
its transposition. $\mathbf{e}_k$ denotes the $k$th column vector of the identity matrix.
Other notations will be introduced when needed.

\section{Problem Setting}\label{sec:Fund_problem}

The settings of sparse estimation with general perturbation is introduced in this section. In the case of general perturbation, the measurement vector is observed via a corrupted sensing matrix as
\begin{equation}\label{fund_eq}
\mathbf{y}=(\mathbf{A}+\mathbf{E})\mathbf{x}+\mathbf{n},
\end{equation}
where $\mathbf{x}$ is the deterministic parameter to be estimated, and $\mathbf{y}$ is the measurement vector. $\mathbf{E}\in\mathbb{R}^{m\times n}$ represents the perturbation on the sensing matrix, whose elements are i.i.d. Gaussian distributed random variables with zero mean and variance $\sigma_e^2$. The vector $\mathbf{n}\sim\mathcal{N}(0,\sigma_n^2\mathbf{U}_m)$ is the noise on the measurement vector $\mathbf{y}$, and is independent of $\mathbf{E}$.

The parameter $\mathbf{x}$ is supposed to be sparse, i.e. the size of its support is far less than its dimension.
The support of $\mathbf{x}$ is denoted by $S$, and its size is assumed to satisfy $|S|=\|\mathbf{x}\|_{\ell_0}\leq s$.
Furthermore, it is adopted in the following text that
\begin{equation}\label{eq:spark}
\mathrm{spark}(\mathbf{A})>2s ,
\end{equation}
where $\mathrm{spark}(\mathbf{A})$ is defined as the smallest possible number $k$ such that there exists a subgroup of $k$ columns from $\mathbf{A}$ that are linearly dependent \cite{Donoho_OpGen}. The above prerequisite ensures that two different $s$-sparse signals will not share the same measurement vector if the measurement is precise.

An estimator $\hat{\mathbf{x}}=\hat{\mathbf{x}}(\mathbf{y})$ is a function of the measurement vector, and is essentially a random variable.
A widely used criterion of the performance of an estimator is the mean square error (MSE), given by
\begin{equation}
\mathrm{mse}(\hat{\mathbf{x}})=E_{\mathbf{y};\mathbf{x}}[\|\hat{\mathbf{x}}(\mathbf{y})-\mathbf{x}\|_{\ell_2}^2].
\end{equation}
Here, $E_{\mathbf{y};\mathbf{x}}[\cdot]$ denotes the expectation taken with respect to the pdf $p(\mathbf{y};\mathbf{x})$ of the measurement $\mathbf{y}$ parameterized by $\mathbf{x}$. Note that the MSE is in general dependent on $\mathbf{x}$.

In this paper only unbiased estimators are considered. Unbiased estimators are the ones that satisfy
\begin{equation}
E_{\mathbf{y};\mathbf{x}}[\hat{\mathbf{x}}(\mathbf{y})]=\mathbf{x},\quad \forall\mathbf{x}\in \mathcal{X}.
\end{equation}
Here $\mathcal{X}$ denotes the set of all possible values of the parameter $\mathbf{x}$. In the sparse setting, the notation $\mathcal{X}_s$ is used for this set and could be formulated as
\begin{equation}
\mathcal{X}_s=\{\mathbf{x}\in\mathbb{R}^n:\|\mathbf{x}\|_{\ell_0}\leq s\}.
\end{equation}
The set of all unbiased estimators will be denoted by $\mathcal{U}$.
For every unbiased estimator, its MSE at a specific parameter value possesses a lower bound known as the Barankin bound (BB) \cite{BarabkinBound}. Unfortunately, the BB often does not possess a closed-form expression, or its computation is of great complexity.
In the remainder of this paper, two types of lower bounds of the BB, the constrained Cram\'{e}r-Rao bound (CCRB) and the Hammersley-Chapman-Robbins bound (HCRB), are discussed for sparse estimation with general perturbation. As they are lower bounds of the BB, they can also be viewed as the lower bounds of the MSE of unbiased estimators. Although they are not as tight as the BB, they usually possess simpler expressions and can provide insights into the properties of the BB.

\section{The Constrained CRB}\label{sec:CCRB}

In this section, the constrained Cram\'{e}r-Rao bound (CCRB) of the estimation problem \eqref{fund_eq} is considered. The CCRB generalizes the original CRB to the case where the parameter is constrained in an arbitrary given set. Researches on CCRB have been developed recently and especially on the situation of sparse estimation \cite{Eldar_SingFIM,Eldar_SparseCRB}. The CCRB can be summarized by the following proposition.

\begin{prop}\label{prop_CCRB}\cite{Eldar_SingFIM}
Suppose that the parameter $\mathbf{x}\in\mathbb{R}^n$ lies in a given set $\mathcal{X}$\footnote{There are certain requirements that $\mathcal{X}$ has to meet. Refer to \cite{Eldar_SparseCRB} for detailed exposition. Fortunately, the set $\mathcal{X}_s$ of the sparse setting meets all the requirements.}, and $\mathbf{x}_0$ is a specific value of $\mathbf{x}$. Define the set $\mathcal{F}(\mathbf{x}_0)$ as follows,
\begin{equation*}
\begin{aligned}
\mathcal{F}(\mathbf{x}_0) = \{\mathbf{v}\in\mathbb{R}^n: \exists~\epsilon_0(\mathbf{v})>0 
\ \mathrm{s.t.}~\forall\epsilon\in(0,\epsilon_0(\mathbf{v})),\mathbf{x}_0+\epsilon\mathbf{v}\in \mathcal{X}\}.
\end{aligned}
\end{equation*}
It can be proved that $\mathcal{F}(\mathbf{x}_0)$ is a subspace of $\mathbb{R}^n$. Let $\mathbf{V}=[\mathbf{v}_1,\ldots,\mathbf{v}_l]$ be an orthogonal basis of $\mathcal{F}(\mathbf{x}_0)$, and $\mathbf{J}$ be the Fisher information matrix (FIM),
\begin{equation}
\mathbf{J}(\mathbf{x}_0)=
E_{\mathbf{y};\mathbf{x}_0}\left[\big(\nabla_{\mathbf{x}}\ln
p(\mathbf{y};\mathbf{x})\big)\big(\nabla_{\mathbf{x}}^\mathrm{T}\ln
p(\mathbf{y};\mathbf{x})\big)\right].
\end{equation}
If
\begin{equation}\label{CCRB_condition}
\mathcal{R}(\mathbf{V}\mathbf{V}^\mathrm{T})\subseteq
\mathcal{R}(\mathbf{V}\mathbf{V}^\mathrm{T}\mathbf{J}\mathbf{V}\mathbf{V}^\mathrm{T}),
\end{equation}
where $\mathcal{R}(\mathbf{P})$ is the column space of $\mathbf{P}$, then for any estimator $\hat{\mathbf{x}}$
which is unbiased in the neighborhood of $\mathbf{x}_0$, its covariance matrix $\mathrm{Cov}(\hat{\mathbf{x}})$ satisfies
\begin{equation}\label{Cov_CCRB_general}
\mathrm{Cov}(\hat{\mathbf{x}})\succeq
\mathbf{V}(\mathbf{V}^\mathrm{T}\mathbf{J}\mathbf{V})^{\dagger}\mathbf{V}^\mathrm{T},
\end{equation}
where $\mathbf{P}\succeq\mathbf{Q}$ means that $\mathbf{P}-\mathbf{Q}$ is positive semidefinite. The trace of the covariance matrix gives the MSE of the estimator. Conversely, if \eqref{CCRB_condition} does not hold, then there exists no finite variance estimator which is unbiased in the neighborhood of $\mathbf{x}_0$.
\end{prop}

\begin{rem}\blue{
Note that the estimators considered in Proposition \ref{prop_CCRB} are ``unbiased in the neighborhood of $\mathbf{x}_0$'', which can be rigorously formulated as follows: define $\mathbf{b}(\mathbf{x})=E[\hat{\mathbf{x}}-\mathbf{x}]$ to be the bias at $\mathbf{x}$, then one says that the estimator $\hat{\mathbf{x}}$ is unbiased in the neighborhood of $\mathbf{x}_0$ if and only if
\begin{equation}
\forall\mathbf{v}\in\mathcal{F}(\mathbf{x}_0),\quad
\mathbf{b}(\mathbf{x}_0)=0~\textrm{and}~\left.\frac{\partial \mathbf{b}(\mathbf{x})}{\partial \mathbf{v}}\right|_{\mathbf{x}_0}=0.
\end{equation}}
We denote the set containing all the estimators unbiased in the neighborhood of $\mathbf{x}_0$ as $\mathcal{U}_{\mathbf{x}_0}$. In the sparse setting, it can be seen that
\begin{equation}
\mathcal{U}\subset\mathcal{U}_{\mathbf{x}_0},\quad\forall\mathbf{x}_0\in\mathcal{X}_s,
\end{equation}
therefore the CCRB is certainly lower than the BB. Nevertheless, the CCRB has simple closed-form expression which is convenient to analyze. In this section we relax our restrictions on the estimators to be unbiased in the neighborhood of a specific parameter value.
\end{rem}

From Proposition \ref{prop_CCRB}, it can be seen that the computation of CCRB mainly relies on the computation of the FIM $\mathbf{J}$ and the orthogonal basis $\mathbf{V}$.
The FIM $\mathbf{J}$ is given by the following lemma.

\begin{lem}\label{lemma:FIM}
The Fisher information matrix is given by
\begin{equation}\label{FIM_xS}
\mathbf{J}(\mathbf{x})=\frac{1}{\sigma_{\mathbf{x}}^2}\left[\mathbf{A}^\mathrm{T}\mathbf{A}
+2m\sigma_e^4\frac{\mathbf{x}\mathbf{x}^\mathrm{T}}{\sigma_{\mathbf{x}}^2}\right],
\end{equation}
where $\sigma_{\mathbf{x}}^2$ is defined as
\begin{equation}\label{equiv_noise}
\sigma_{\mathbf{x}}^2=\sigma_e^2\|\mathbf{x}\|_{\ell_2}^2+\sigma_n^2.
\end{equation}
\end{lem}
\begin{proof}
The proof is postponed to Appendix \ref{proof_FIM}.
\end{proof}

Next we deal with the orthogonal basis $\mathbf{V}$ of the subspace $\mathcal{F}$. The cases in which $\|\mathbf{x}\|_{\ell_0}=s$ and $\|\mathbf{x}\|_{\ell_0}<s$ should be discussed separately. For the case $\|\mathbf{x}\|_{\ell_0}=s$, it can be seen that for every $k\in S=\mathrm{supp}(\mathbf{x})$, one has $\|\mathbf{x}+\epsilon\mathbf{e}_k\|_{\ell_0}\leq s$, i.e. $\mathbf{x}+\epsilon\mathbf{e}_k\in\mathcal{X}_s$ for arbitrary $\epsilon$, and therefore $\mathbf{e}_k\in\mathcal{F}$; on the other hand, for $k\notin S$,  one has $\|\mathbf{x}+\epsilon\mathbf{e}_k\|_{\ell_0}> s$ and thus $\mathbf{e}_k\notin\mathcal{F}$. It follows that the subspace $\mathcal{F}(\mathbf{x})$ can be formulated as
\begin{equation}
\mathcal{F}(\mathbf{x})=\mathrm{span}(\{\mathbf{e}_{S_1},\ldots,\mathbf{e}_{S_s}\}),
\end{equation}
in which $S_1,\ldots,S_s$ are the elements of the support $S=\mathrm{supp}(\mathbf{x})$, and the basis $\mathbf{V}$ can take the following form
\begin{equation}\label{V_full_support}
\mathbf{V}=[\mathbf{e}_{S_1},\ldots,\mathbf{e}_{S_s}].
\end{equation}

For the case $\|\mathbf{x}\|_{\ell_0}<s$, the situation is rather different, because for every $\mathbf{e}_k,~k=1,\ldots,n$,  one has $\|\mathbf{x}+\epsilon\mathbf{e}_k\|_{\ell_0}\leq s$. Thus it can be concluded that $\mathcal{F}(\mathbf{x})=\mathbb{R}^n$ in this case, and the basis $\mathbf{V}$ can be given by $[\mathbf{e}_1,\ldots,\mathbf{e}_n]=\mathbf{U}_n$.

With the form of the FIM $\mathbf{J}$ and the basis $\mathbf{V}$, one can readily derive the CCRB of the problem \eqref{fund_eq}. The situation in which $\mathbf{x}$ has maximal support ($\|\mathbf{x}\|_{\ell_0}=s$) is first analyzed.

\begin{thm}\label{thm:CCRB}
For $\|\mathbf{x}\|_{\ell_0}=s$, the CCRB is given by
\begin{equation}\label{CCRB_1}
\begin{aligned}
\sigma_{\mathbf{x}}^2\Bigg[&\tr\left((\mathbf{A}_S^\mathrm{T}\mathbf{A}_S)^{-1}\right) -\frac{2m\sigma_e^4\|(\mathbf{A}_S^\mathrm{T}\mathbf{A}_S)^{-1}\mathbf{x}_S\|^2_{\ell_2}} {\sigma_{\mathbf{x}}^2+2m\sigma_e^4\mathbf{x}_S^\mathrm{T}(\mathbf{A}_S^\mathrm{T}\mathbf{A}_S)^{-1}\mathbf{x}_S}\Bigg],
\end{aligned}
\end{equation}
where $\sigma_{\mathbf{x}}^2$ is defined in \eqref{equiv_noise}.
\end{thm}

\begin{proof}
The condition \eqref{CCRB_condition} should be checked first. The matrix $\mathbf{V}^\mathrm{T}\mathbf{J}\mathbf{V}$ is given by
\begin{equation*}
\begin{aligned}
\mathbf{V}^\mathrm{T}\mathbf{J}\mathbf{V}
&=\mathbf{V}^\mathrm{T}\frac{1}{\sigma_{\mathbf{x}}^2}\left[\mathbf{A}^\mathrm{T}\mathbf{A}
+2m\sigma_e^4\frac{\mathbf{x}\mathbf{x}^\mathrm{T}}{\sigma_{\mathbf{x}}^2}\right]\mathbf{V} \\
&=\frac{1}{\sigma_{\mathbf{x}}^2}\left[\mathbf{A}_S^\mathrm{T}\mathbf{A}_S+
2m\sigma_e^4\frac{\mathbf{x}_S\mathbf{x}_S^\mathrm{T}}{\sigma_{\mathbf{x}}^2}\right].
\end{aligned}
\end{equation*}
Because we have assumed that $\mathrm{spark}(\mathbf{A})>2s$, it follows that $\mathbf{A}_S$ has full column rank, and thus $\mathbf{V}^\mathrm{T}\mathbf{J}\mathbf{V}$ is invertible by employing the Sherman-Morrison formula \cite{Golub_MatComp}, which gives
\begin{equation}\label{VJV_inv}
\begin{aligned}
\left(\mathbf{V}^\mathrm{T}\mathbf{J}\mathbf{V}\right)^{-1}
=\sigma_\mathbf{x}^2\left[(\mathbf{A}_S^\mathrm{T}\mathbf{A}_S)^{-1}-\frac{2m\sigma_e^4(\mathbf{A}_S^\mathrm{T}\mathbf{A}_S)^{-1}\mathbf{x}_S\mathbf{x}_S^\mathrm{T}(\mathbf{A}_S^\mathrm{T}\mathbf{A}_S)^{-1}}
{\sigma_{\mathbf{x}}^2+2m\sigma_e^4\mathbf{x}_S^\mathrm{T}(\mathbf{A}_S^\mathrm{T}\mathbf{A}_S)^{-1}\mathbf{x}_S}\right],
\end{aligned}
\end{equation}
and thus $\mathcal{R}(\mathbf{V}\mathbf{V}^\mathrm{T}\mathbf{J}\mathbf{V}\mathbf{V}^\mathrm{T}) =\mathcal{R}(\mathbf{V}\mathbf{V}^\mathrm{T})$, i.e. the CCRB exists. The expression of the CCRB can also be obtained from \eqref{VJV_inv} that
\begin{equation*}
\begin{aligned}
\mathrm{mse}(\hat{\mathbf{x}})\geq &~\tr\left(\mathbf{V}(\mathbf{V}^\mathrm{T}\mathbf{J}\mathbf{V})^{-1}\mathbf{V}^\mathrm{T}\right) \\
=&~\tr\left(\mathbf{V}^\mathrm{T}\mathbf{V}(\mathbf{V}^\mathrm{T}\mathbf{J}\mathbf{V})^{-1}\right) \\ =&~ \tr\left((\mathbf{V}^\mathrm{T}\mathbf{J}\mathbf{V})^{-1}\right) \\
=&~\sigma_\mathbf{x}^2\Bigg[\tr\left((\mathbf{A}_S^\mathrm{T}\mathbf{A}_S)^{-1}\right)-\frac{2m\sigma_e^4\|(\mathbf{A}_S^\mathrm{T}\mathbf{A}_S)^{-1}\mathbf{x}_S\|^2_{\ell_2}}
{\sigma_\mathbf{x}^2+2m\sigma_e^4\mathbf{x}_S^\mathrm{T}(\mathbf{A}_S^\mathrm{T}\mathbf{A}_S)^{-1}\mathbf{x}_S}\Bigg].
\end{aligned}
\end{equation*}

\end{proof}

Next consider the case in which $\mathbf{x}$ has nonmaximal support. The CCRB of this case can be summarized as the following theorem.

\begin{thm}\label{CCRB_nonmaximal}
For $\|\mathbf{x}\|_{\ell_0}<s$, if the FIM $\mathbf{J}$ is nonsingular, then the CCRB exists. Furthermore,
if $\mathbf{A}$ has full column rank, then the CCRB is given by
\begin{equation}\label{CCRB_2}
\sigma_{\mathbf{x}}^2\left[\tr\left((\mathbf{A}^\mathrm{T}\mathbf{A})^{-1}\right) -\frac{2m\sigma_e^4\|(\mathbf{A}^\mathrm{T}\mathbf{A})^{-1}\mathbf{x}\|^2_{\ell_2}}
{\sigma_{\mathbf{x}}^2+2m\sigma_e^4\mathbf{x}^\mathrm{T}(\mathbf{A}^\mathrm{T}\mathbf{A})^{-1}\mathbf{x}}
\right].
\end{equation}
If the FIM $\mathbf{J}$ is singular, then there do not exist finite variance estimators that are unbiased in the neighborhood of $\mathbf{x}$.
\end{thm}
\begin{proof}
Because in this case $\mathbf{V}=\mathbf{U}_n$, the two subspaces are respectively $\mathcal{R}(\mathbf{V}\mathbf{V}^\mathrm{T})=\mathbb{R}^n$ and $\mathcal{R}(\mathbf{V}\mathbf{V}^\mathrm{T}\mathbf{J}\mathbf{V}\mathbf{V}^\mathrm{T})=\mathcal{R}(\mathbf{J})$.
Therefore when $\mathbf{J}$ is invertible, it can be seen that the condition \eqref{CCRB_condition} holds, and the CCRB can be obtained by taking the trace of $\mathbf{J}^{-1}$. In the special case where $\mathbf{A}$ has full column rank, the inverse $\mathbf{J}^{-1}$ can be calculated with the help of the the
Sherman-Morrison formula \cite{Golub_MatComp},
\begin{equation}
\mathbf{J}^{-1}
=\sigma_\mathbf{x}^2\left[(\mathbf{A}^\mathrm{T}\mathbf{A})^{-1}-\frac{2m\sigma_e^4(\mathbf{A}^\mathrm{T}\mathbf{A})^{-1}\mathbf{x}\mathbf{x}^\mathrm{T}(\mathbf{A}^\mathrm{T}\mathbf{A})^{-1}}
{\sigma_{\mathbf{x}}^2+2m\sigma_e^4\mathbf{x}^\mathrm{T}(\mathbf{A}^\mathrm{T}\mathbf{A})^{-1}\mathbf{x}}\right],
\end{equation}
and the CCRB is
\begin{equation}
\begin{aligned}
\mathrm{mse}(\hat{\mathbf{x}})\geq\tr\left(\mathbf{J}^{-1}\right)
=\sigma_{\mathbf{x}}^2\Bigg[\tr\left((\mathbf{A}^\mathrm{T}\mathbf{A})^{-1}\right)-\frac{2m\sigma_e^4\|(\mathbf{A}^\mathrm{T}\mathbf{A})^{-1}\mathbf{x}\|^2_{\ell_2}}
{\sigma_{\mathbf{x}}^2+2m\sigma_e^4\mathbf{x}^\mathrm{T}(\mathbf{A}^\mathrm{T}\mathbf{A})^{-1}\mathbf{x}}\Bigg].
\end{aligned}
\end{equation}

When $\mathbf{J}$ is not invertible, the dimension of the column space of $\mathbf{J}$ is less than $n$. Thus the condition \eqref{CCRB_condition} does not hold, and estimators that are unbiased in the neighborhood of $\mathbf{x}$ do not exist.
\end{proof}

Theorem \ref{CCRB_nonmaximal} illustrates that for the nonmaximal support case, the prior information of sparsity cannot lower the theoretical bound of  estimation error compared to the ordinary problem where $\mathbf{x}$ can be any vector in $\mathbb{R}^n$. This demonstrates a gap between the maximal and nonmaximal support cases of the CCRB, which is the main topic of the next subsection.

\subsection{Gap between the Maximal and the Nonmaximal Cases}

The gap between the maximal and nonmaximal cases of the CCRB can be revealed from the following example.
Suppose we observe a sparse vector $\mathbf{x}\in\mathcal{X}_s$ with $s$ nonzero entries. Then by the result of Theorem \ref{thm:CCRB},
the CCRB is given by \eqref{CCRB_1}. Next we assume that one of the nonzero components, say $x_q$, tends to zero. Consequently, the CCRB given by
\eqref{CCRB_1} tends to a specific limit $\gamma_1$. However, when $x_q$ equals zero, the CCRB cannot be computed by \eqref{CCRB_1} anymore
because its support is now nonmaximal, therefore the CCRB of the $x_q=0$ case is given by \eqref{CCRB_2}, which we temporarily denote as
$\gamma_2$. It is interesting to find that $\gamma_1$ and $\gamma_2$ are not equal to each other, which means that the CCRB is not a continuous function of the parameter $\mathbf{x}$.

Generally one has $\gamma_1\leq\gamma_2$, which can be inferred as follows. $\gamma_1$ could be seen as the CRB of estimators which are unbiased
on the subspace $\mathrm{span}(\{\mathbf{e}_i:i\in\mathrm{supp}(\mathbf{x})\})$, while $\gamma_2$ could be seen as the CRB of estimators unbiased
on $\mathbb{R}^n$. If the former class of estimators is denoted by $\mathcal{U}_1$, and the latter is denoted by $\mathcal{U}_2$, it can be
seen that $\mathcal{U}_1\supset\mathcal{U}_2$, and thus the lower bounds of estimation error of the two classes should satisfy
$\gamma_1\leq\gamma_2$. This conclusion can also be verified by numerical approaches.

This gap originates from the ``discontinuity'' of the restriction that the estimator should be unbiased in the neighborhood of a specific parameter value. The ``neighborhood'' of a parameter point having maximal support in $\mathcal{X}_s$ has an entirely different structure from that of a parameter point having nonmaximal support: the former is a subspace that is locally identical to $\mathbb{R}^s$, while the latter is a union of $s$-dimensional subspaces. Fig. \ref{fig_CCRB_discont} is a geometric illustration of the structure of the neighborhood of $\mathbf{x}$. It can be seen that as $x_q\rightarrow0$, the structure of the neighborhood of $\mathbf{x}$ will have an abrupt change: from being locally identical to $\mathbb{R}^s$ to being locally identical to $\mathbb{R}^n$. This is the cause of the gap between the maximal and nonmaximal cases.

\begin{figure}[!t]
\centerline{
\subfigure[The maximal support case]{\includegraphics[width=.25\textwidth]{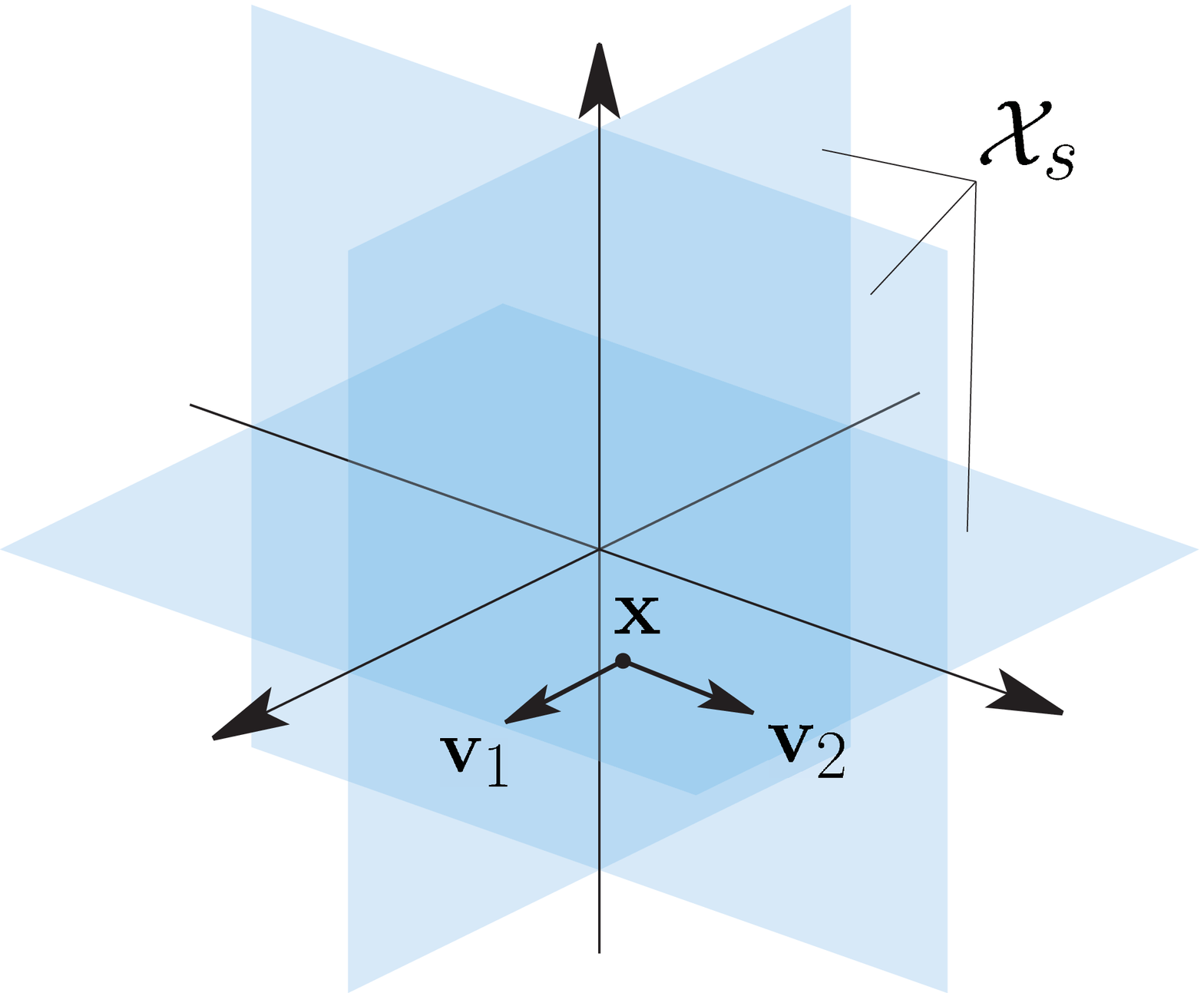}
\label{fig_CCRB_maxsupp}}
\hfil
\subfigure[The nonmaximal support case]{\includegraphics[width=.25\textwidth]{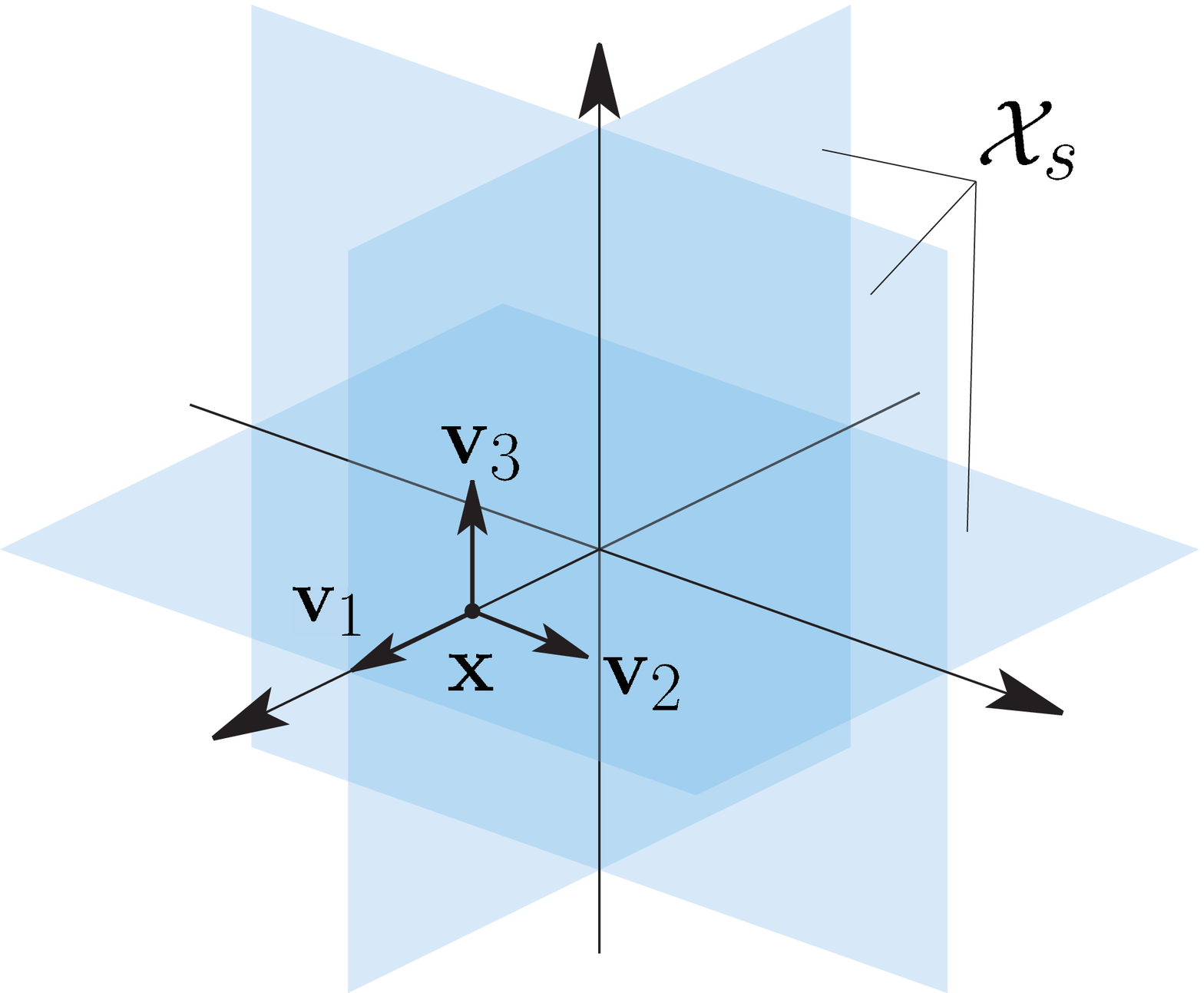}
\label{fig_CCRB_nonmaxsupp}}}
\caption{A geometrical demonstration of the discontinuity between the neighborhood structures of the maximal and the nonmaximal support cases.
The vectors $\mathbf{v}_i$ are base vectors of the neighborhood subspace $\mathcal{F}$.}
\label{fig_CCRB_discont}
\end{figure}

On the other hand, if a stronger condition, global unbiasedness, is imposed on the considered estimators instead of unbiasedness in the neighborhood, i.e. the estimators are restricted to be unbiased for all $\mathbf{x}\in\mathcal{X}_s$, then this discontinuity should not occur. Thus the corresponding lower bound should also be continuous as $x_q\rightarrow0$, i.e. $\mathbf{x}$ changes from having maximal support to having nonmaximal support. Since $\gamma_1\leq\gamma_2$, it is further demonstrated that the CCRB for maximal support is not sufficiently tight for estimators in $\mathcal{U}$, especially when the support of $\mathbf{x}$ is nearly nonmaximal, i.e. at least one of the non-zero entries is small compared to other non-zero entries.

\subsection{Further Analysis of the CCRB} \label{analysis_CCRB}

In this subsection we aim to analyze the behavior of the CCRB. The following analysis will mainly focus on the maximal support case. As can be seen from \eqref{CCRB_1} and \eqref{CCRB_2}, the CCRB of the two cases share a similar form, and thus the main results provided in the following will still be valid for the nonmaximal support case with minor modifications.

The expression \eqref{CCRB_1} contains two terms, the first of which is rather simple, while the second of which is much more complicated.
The following proposition relates the first term $\sigma_{\mathbf{x}}^2\tr\left((\mathbf{A}_S^\mathrm{T}\mathbf{A}_S)^{-1}\right)$ to the MSE of the oracle estimator, which provides an intuitive explanation. The definition and performance of the oracle estimator can be summarized as follows.
\begin{prop}\label{prop:pinv}
For a given support $S$ whose size is $s$, the oracle estimator $\hat{\mathbf{x}}_{\mathrm{or}}$ is defined as
\begin{equation*}
\begin{aligned}
\left(\hat{\mathbf{x}}_{\mathrm{or}}(\mathbf{y})\right)_S
&=\mathbf{A}_S^\dagger\mathbf{y}=(\mathbf{A}_S^\mathrm{T}\mathbf{A}_S)^{-1}\mathbf{A}_S^\mathrm{T}\mathbf{y}, \\
\left(\hat{\mathbf{x}}_{\mathrm{or}}(\mathbf{y})\right)_{S^c}&=0.
\end{aligned}
\end{equation*}
This estimator is unbiased in the neighborhood of any parameter value whose support is $S$.
The MSE of $\hat{\mathbf{x}}_{\mathrm{or}}$ is
\begin{equation}
\sigma_{\mathbf{x}}^2\tr\left((\mathbf{A}_S^\mathrm{T}\mathbf{A}_S)^{-1}\right),\qquad\forall\mathbf{x}:\mathrm{supp}(\mathbf{x})=S.
\end{equation}
\end{prop}

Proposition \ref{prop:pinv} demonstrates that the first term of the CCRB is just the MSE of the oracle estimator. This term is also similar to the CCRB of the case where only measurement noise exists. Various references (e.g. \cite{Candes_Dantzig,Eldar_SparseCRB}) have shown that the CCRB with only measurement noise is the variance of the noise on $\mathbf{y}$ multiplied by $\tr{(\mathbf{A}_S^\mathrm{T}\mathbf{A}_S)^{-1}}$, and the oracle estimator achieves this bound.

The second term of \eqref{CCRB_1} stems from the dependence of the variance of the total noise on the parameter $\mathbf{x}$, and reveals a possibility that utilizing matrix perturbation might help estimate $\mathbf{x}$ more accurately\blue{\footnote{\blue{
An extreme example is given in Appendix \ref{eg_analy_CCRB}.}}}. However, it can be shown that this term is not dominant in the CCRB under certain assumptions on $\mathbf{A}$, which is given as follows.

\begin{assump}\label{assump:sensing_mat}
For the sensing matrix $\mathbf{A}$ in \eqref{CCRB_1}, it is assumed that there exist constants $\vartheta_{\mathrm{l},s}\in(0,1)$ and $\vartheta_{\mathrm{u},s}>0$ such that
\begin{equation}
(1-\vartheta_{\mathrm{l},s})\|\mathbf{x}\|^2_{\ell_2}\leq\|\mathbf{A}\mathbf{x}\|_{\ell_2}^2\leq(1+\vartheta_{\mathrm{u},s})\|\mathbf{x}\|_{\ell_2}^2
\end{equation}
for any $s$-sparse parameter $\mathbf{x}$.
\end{assump}

Assumption \ref{assump:sensing_mat} is very similar to the restricted isometry property \cite{Candes_RIP,Candes_DecodingLP}, and has the same form as the asymmetric restricted isometry property in \cite{Blanchard_RIP}. However, the sensing matrix is not restricted to be underdetermined in our case. With Assumption \ref{assump:sensing_mat}, the following theorem can be derived providing lower and upper bounds on the second term of the CCRB.

\begin{thm}\label{thm:bounds_dCCRB}
Denote the opposite of the second term of the CCRB as $d_\mathrm{CCRB}$, i.e.
\begin{equation}\label{dCCRB}
d_\mathrm{CCRB}
=\sigma_\mathbf{x}^2\cdot\frac{2m\sigma_e^4\|(\mathbf{A}_S^\mathrm{T}\mathbf{A}_S)^{-1}\mathbf{x}_S\|^2_{\ell_2}}
{\sigma_\mathbf{x}^2+2m\sigma_e^4\mathbf{x}_S^\mathrm{T}(\mathbf{A}_S^\mathrm{T}\mathbf{A}_S)^{-1}\mathbf{x}_S}.
\end{equation}
Then $d_\mathrm{CCRB}$ satisfies the following inequalities:
\begin{equation}\label{ubound_dCCRB}
\begin{aligned}
d_\mathrm{CCRB}\lesseqqgtr
\sigma_\mathbf{x}^2\frac{(1+\vartheta_{\pm,s})^2}{(1+\vartheta_{\mp,s})^2} \cdot\frac{2c_e}{2(1+\vartheta_{\mp,s})c_e+1+\vartheta_{\pm,s}+c_n/c_e},
\end{aligned}
\end{equation}
where $\vartheta_{+,s}=\vartheta_{\mathrm{u},s},\vartheta_{-,s}=-\vartheta_{\mathrm{l},s}$, and
\begin{equation*}
\begin{aligned}
c_e=\frac{ms\sigma_e^2}{\sum_{i=1}^m\sum_{j\in S}A_{ij}^2}=\frac{ms\sigma_e^2}{\tr(\mathbf{A}_S^\mathrm{T}\mathbf{A}_S)}, \quad
c_n=\frac{m\sigma_n^2}{\|\mathbf{x}\|_{\ell_2}^2}
\end{aligned}
\end{equation*}
indicate the matrix perturbation level and the measurement noise level\footnote{One may argue that $m\sigma_n^2/\|\mathbf{A}\mathbf{x}\|_{\ell_2}^2$ seems more reasonable as an indication of the measurement noise level. However, because of Assumption \ref{assump:sensing_mat}, one has $\|\mathbf{A}\mathbf{x}\|_{\ell_2}^2\approx\|\mathbf{x}\|_{\ell_2}^2$, and thus for simplicity $m\sigma_n^2/\|\mathbf{x}\|_{\ell_2}^2$ is employed instead.} respectively. The ratio $\gamma_\mathrm{CCRB}=d_\mathrm{CCRB}/(\mathrm{CCRB}+d_\mathrm{CCRB})$, i.e. the ratio of the second term to the first term of the CCRB, is bounded by
the following inequalities:
\begin{equation}\label{ubound_dratio}
\begin{aligned}
\gamma_\mathrm{CCRB}
\lesseqqgtr\frac{1}{s}\frac{(1+\vartheta_{\pm,s})^3}{(1+\vartheta_{\mp,s})^2} \cdot\frac{2c_e}{2(1+\vartheta_{\mp,s})c_e+1+\vartheta_{\pm,s}+c_n/c_e}
\end{aligned}
\end{equation}
\end{thm}
\begin{proof}
See Appendix \ref{proof_bounds_dCCRB}.
\end{proof}

Theorem \ref{thm:bounds_dCCRB} provides very simple approximate expressions that captures how $d_\mathrm{CCRB}$ and $\gamma_\mathrm{CCRB}$ vary with the noise level
$c_e$ and $c_n$:
\begin{align}
d_\mathrm{CCRB}&\approx\sigma_\mathbf{x}^2\cdot\frac{2c_e}{2c_e+1+c_n/c_e}, \label{approx_dCCRB}\\
\gamma_\mathrm{CCRB}&\approx\frac{1}{s}\cdot\frac{2c_e}{2c_e+1+c_n/c_e}, \label{approx_dratio}
\end{align}
provided that the constants $\vartheta_{\mathrm{l},s}$ and $\vartheta_{\mathrm{u},s}$ are small compared to $1$.
Apparently the quantity $2c_e/(2c_e+1+c_n/c_e)$ is always less than $1$, and therefore $\gamma_\mathrm{CCRB}$ can be upper-bounded approximately by $1/s$. Furthermore, \eqref{approx_dratio} implies that as $s$ increases, the second term $d_\mathrm{CCRB}$ becomes less and less important, and finally becomes negligible. This can be considered as the asymptotic behavior of $d_\mathrm{CCRB}$ or the CCRB,
which can be summarized as the following corollary.

\begin{corol}\label{corol:asymp}
Assume that there exist constants $\epsilon_\mathrm{l}\in(0,1)$ and $\epsilon_\mathrm{u}>0$ such that when $s$ tends to infinity, the constants $\vartheta_{\mathrm{l},s}$ and $\vartheta_{\mathrm{u},s}$ keep satisfying $\vartheta_{\mathrm{l},s}<\epsilon_\mathrm{l}$ and $\vartheta_{\mathrm{u},s}<\epsilon_\mathrm{u}$ for every $s$. Then if $c_e$ and $c_n$ remains constant, $\gamma_\mathrm{CCRB}$ possesses the following asymptotic behavior:
\begin{equation}
\frac{A}{s}\leq\gamma_\mathrm{CCRB}\leq\frac{B}{s},
\end{equation}
where $A$ and $B$ are some positive constants.
\end{corol}

\begin{rem}
This corollary demonstrates that as $s\rightarrow\infty$, the CCRB approaches $\sigma_\mathbf{x}^2\tr\left((\mathbf{A}_S^\mathrm{T}\mathbf{A}_S)^{-1}\right)$ which is just the MSE of the oracle estimator. The oracle estimator is the solution of the following minimization problem:
\begin{equation*}
\mathop{\mathrm{arg}}_{\mathbf{x}_S}\min\|\mathbf{y}-\mathbf{A}_S\mathbf{x}_S\|_{\ell_2}^2,
\end{equation*}
which merely minimizes the residual and completely ignores the dependence of the total noise on $\mathbf{x}$. It can be concluded that as $s$ increases, less information can be possibly obtained from the dependence of the total noise on $\mathbf{x}$ to help reduce the estimation error, and finally this dependence can be ignored.
\end{rem}

\section{The Hammersley-Chapman-Robbins Bound}\label{sec:HCRB}

The constrained CRB given in the above section possesses a simple closed form, but it takes into account only the unbiasedness in the neighborhood of a specific parameter value rather than the global unbiasedness for all sparse vectors.
Therefore it can be anticipated that for estimators that are globally unbiased for sparse parameter values (i.e. for all estimators in
$\mathcal{U}$), the CCRB is not a very tight lower bound.

In this section, the Hammersley-Chapman-Robbins bound (HCRB) is derived for sparse estimation under the setting of general perturbation.
However, the calculation of the HCRB for general sensing matrix $\mathbf{A}$ is much more complicated, and therefore attention is focussed only on the simplest case of unit sensing matrix in this section \cite{Eldar_SSNM,Eldar_SLM_RKHS}. Nevertheless, the HCRB of this simple case is still instructive for us to have a qualitative understanding of the HCRB for general cases.

The HCRB in the context of sparse estimation with general perturbation can be summarized as the following lemma.

\begin{lem}\label{lem:HCRB_general}
In the setting given in Section \ref{sec:Fund_problem}, consider a specific parameter value $\mathbf{x}\in\mathcal{X}_s$.
Suppose $\{\mathbf{v}_i\}_{i=1}^k$ are $k$ vectors such that $\mathbf{x}+\mathbf{v}_i\in\mathcal{X}_s$ for all $i=1,\ldots,k$.
Then the covariance matrix of any unbiased estimator $\hat{\mathbf{x}}\in\mathcal{U}$ at $\mathbf{x}$ satisfies
\begin{equation}\label{HCRB_sparse_general}
\mathrm{Cov}(\hat{\mathbf{x}})\succeq
\mathbf{V}\mathbf{H}^{\dagger}\mathbf{V}^\mathrm{T}.
\end{equation}
Here the matrix $\mathbf{V}$ is given by
\begin{equation}
\mathbf{V}=[\mathbf{v}_1,\ldots,\mathbf{v}_k],
\end{equation}
and the $(i,j)$th element of $\mathbf{H}$ is
\begin{equation}\label{Hij_general}
\begin{aligned}
H_{ij} = \left(\frac{\sigma_\mathbf{x}^2\varsigma_{\mathbf{x},\mathbf{v}_i,\mathbf{v}_j}^2}{\sigma_{\mathbf{x}+\mathbf{v}_i}^2\sigma_{\mathbf{x}+\mathbf{v}_j}^2}\right)^{\frac{m}{2}}
\exp\left[-\frac{\|\mathbf{A}\mathbf{v}_i\|_{\ell_2}^2}{2\sigma_{\mathbf{x}+\mathbf{v}_i}^2}-\frac{\|\mathbf{A}\mathbf{v}_j\|_{\ell_2}^2}{2\sigma_{\mathbf{x}+\mathbf{v}_j}^2} +\frac{\varsigma_{\mathbf{x},\mathbf{v}_i,\mathbf{v}_j}^2}{2}\left\|\frac{\mathbf{A}\mathbf{v}_i}{\sigma_{\mathbf{x}+\mathbf{v}_i}^2}
+\frac{\mathbf{A}\mathbf{v}_j}{\sigma_{\mathbf{x}+\mathbf{v}_j}^2}\right\|_{\ell_2}^2 \right] - 1,
\end{aligned}
\end{equation}
where
\begin{equation}
\frac{1}{\varsigma_{\mathbf{x},\mathbf{v}_i,\mathbf{v}_j}^2}
=\frac{1}{\sigma_{\mathbf{x}+\mathbf{v}_i}^2}+\frac{1}{\sigma_{\mathbf{x}+\mathbf{v}_j}^2}-\frac{1}{\sigma_{\mathbf{x}}^2}.
\end{equation}
\end{lem}
\begin{proof}
The proof is postponed to Appendix \ref{proof_HCRB_general}.
\end{proof}

It can be seen from Lemma \ref{lem:HCRB_general} that the HCRB is actually a family of lower bounds of unbiased estimators. By employing different sets of $\mathbf{v}_i$, one will generally get different HCRBs, and the tightest one is their supremum. However, it is often impossible to obtain a closed-form expression of the supreme value of the HCRB family, and thus our task is to employ a certain set of $\mathbf{v}_i$ in the hope that the corresponding HCRB will be simple and easy to analyze.

By appropriately choosing a set of $\mathbf{v}_i$ and applying some special techniques, the following theorem of the HCRB\footnote{This lower bound is actually a limit of a family of HCRBs. Nevertheless this bound will still be referred to as the HCRB in the following text.} will be obtained.

\begin{thm}\label{thm:HCRB}
Assume that $n\geq 2$, and denote $\beta=x_q^2/\sigma_\mathbf{x}^2$, where $x_q$ is the smallest entry in magnitude of the parameter
$\mathbf{x}$, with $q$ the corresponding index. Then the MSE of any estimator $\hat{\mathbf{x}}\in\mathcal{U}$ satisfies
\begin{equation}\label{HCRB_closedform}
\begin{aligned}
\mathrm{mse}(\hat{\mathbf{x}})
\geq \sigma_\mathbf{x}^2\left(s-
\frac{2n\sigma_e^4\|\mathbf{x}\|^2_{\ell_2}}
{\sigma_{\mathbf{x}}^2+2n\sigma_e^4\|\mathbf{x}\|^2_{\ell_2}}\right) + \frac{\sigma_\mathbf{x}^2(n-s)\beta{\rm e}^{-\beta}}{{\rm e}^{\beta}-1}\cdot
\left(1-\frac{1}{n-s+{\rm e}^{\beta}(1-g(\beta))^{-1}}\right)
\end{aligned}
\end{equation}
for any $\mathbf{x}\in\mathcal{X}_s$ with $\|\mathbf{x}\|_{\ell_0}=s$.
The specific form of the function $g(\beta)$ is
\begin{equation}\label{form_func_g}
g(\beta)=\frac{\beta(1-2\sigma_e^2\beta)^2}{({\rm e}^{\beta}-1)(1+2n\sigma_e^4\beta)},
\end{equation}
and $g(\beta)$ satisfies
\begin{equation}\label{func_g_property1}
0\leq g(\beta)<1,\quad\forall\beta>0, n\geq2.
\end{equation}
When $\sigma_n$ and $\sigma_e$ are fixed, ones has
\begin{equation}\label{func_g_property2}
\lim_{x_q\rightarrow 0}g(\beta)=1.
\end{equation}
\end{thm}
\begin{proof}
See Appendix \ref{proof_HCRB_thm}.
\end{proof}

The quantity $\beta=x_q^2/\sigma_\mathbf{x}^2$ represents a special kind of signal-to-noise ratio, and can be named the ``worst case entry SNR'' \cite{Eldar_SSNM}. This quantity plays a central role in the transition from maximal support to nonmaximal support. However, it should be noticed that there exists an upper bound of the domain of $\beta$ if $\sigma_e\neq 0$:
\begin{equation}
\begin{aligned}
\beta=\frac{x_q^2}{\sigma_n^2+\sigma_e^2\|\mathbf{x}\|_{\ell_2}^2}\leq\frac{x_q^2}{\sigma_e^2\|\mathbf{x}\|_{\ell_2}^2} =\frac{1}{\sigma_e^2\sum x_i^2/x_q^2}\leq\frac{1}{s\sigma_e^2}.
\end{aligned}
\end{equation}
Therefore the situation here is more complicated than that with only the measurement noise. One of the consequences is that when $x_q\rightarrow+\infty$\blue{\footnote{\blue{
The case $x_q\rightarrow+\infty$ is studied to analyze the large signal-to-measurement-noise-ratio situation. This is for the
convenience of analysis and clarity of exposition.
}}}, the HCRB does not generally tends to the CCRB of maximal support unless $\sigma_e=0$, i.e.
\begin{equation}
\lim_{x_q\rightarrow+\infty}\frac{\mathrm{HCRB}-\mathrm{CCRB}_\textrm{max supp}}{\mathrm{CCRB}_\textrm{max supp}}
>0,\quad\textrm{if }\sigma_e\neq 0.
\end{equation}
This phenomenon reveals that the global unbiasedness has an essential effect on the problem of the $\sigma_e\neq0$ case.
Fortunately, when $\sigma_n$ and $\sigma_e$ remains constant, the HCRB will converge to the CCRB of nonmaximal support as $x_q\rightarrow0$, as can be seen from
\begin{equation}
\lim_{x_q\rightarrow0}\mathrm{HCRB}=\sigma_\mathbf{x}^2
\left(n-
\frac{2n\sigma_e^4\|\mathbf{x}\|^2_{\ell_2}}
{\sigma_{\mathbf{x}}^2+2n\sigma_e^4\|\mathbf{x}\|^2_{\ell_2}}\right).
\end{equation}
Therefore it can be said that the gap between the maximal and nonmaximal cases is eliminated.

We compare the above result with other similar works. In \cite{Eldar_SSNM}, a closed-form expression of the HCRB for $\sigma_e^2=0$ is derived in a similar approach but with different choice of $\{\mathbf{v}_i\}$. Their closed-form result is tighter than ours when $\beta$ is sufficiently large, but is not as tight as ours in the low $\beta$ range and fails to close the gap between maximal and nonmaximal cases. In \cite{Eldar_SLM_RKHS}, another lower bound is provided and is tighter than ours for all $\beta>0$, but the derivation is based on the RKHS formulation of the BB which is difficult to be generalized when $\sigma_e^2\neq0$. Despite that our bound is not the tightest, it still provides a correct qualitative trend of the lower bound of sparse estimation, and is able to deal with matrix perturbation without much effort.

\subsection{Further Analysis of the HCRB}\label{analysis_HCRB}

It has already been mentioned that the domain of $\beta$ possesses an upper bound which is not greater than $1/s\sigma_e^2$, and thus the $x_q\rightarrow+\infty$ limit of the HCRB does not coincide with the CCRB.
This difference implicates the effect of the matrix perturbation on the HCRB and is the major topic of this subsection. For simplicity of analysis, in the following text it is assumed that as $x_q\rightarrow+\infty$, all other components of $\mathbf{x}$ equal $x_q$ asymptotically, which leads to the fact that the upper bound of $\beta$ is $1/s\sigma_e^2$.

As $x_q\rightarrow+\infty$, the HCRB asymptotically equals
\begin{displaymath}
\begin{aligned}
&\sigma_\mathbf{x}^2\left(s-
\frac{2n\sigma_e^4\|\mathbf{x}\|^2_{\ell_2}}
{\sigma_{\mathbf{x}}^2+2n\sigma_e^4\|\mathbf{x}\|^2_{\ell_2}}\right) + \frac{\sigma_\mathbf{x}^2(n-s)
\exp(-1/s\sigma_e^2)}{s\sigma_e^2(\exp(1/s\sigma_e^2)-1)} \\
&\qquad\cdot
\left(1-\frac{1}{n-s+\exp(1/s\sigma_e^2)(1-g(1/s\sigma_e^2))^{-1}}\right).
\end{aligned}
\end{displaymath}
The above expression contains two terms, the first of which is just the CCRB of the maximal support case.
For the second term, it can be seen that the $\sigma_e$'s mainly appear in the exponentials, which demonstrates that there exists a particular value of $\sigma_e$ that separates the low $\sigma_e$ region and the high $\sigma_e$ region. This particular value will be named the transition value of $\sigma_e$ and will be denoted by $\sigma_{e,\mathrm{t}}$.

The analysis and computation of $\sigma_{e,\mathrm{t}}$ is simple. When $\sigma_e^2\ll 1/s$, one has $\exp(1/s\sigma_e^2)\gg 1$ and $\exp(-1/s\sigma_e^2)\approx0$, and thus the HCRB approximates
\begin{displaymath}
\sigma_\mathbf{x}^2
\left(s-
\frac{2n\sigma_e^4\|\mathbf{x}\|^2_{\ell_2}}
{\sigma_{\mathbf{x}}^2+2n\sigma_e^4\|\mathbf{x}\|^2_{\ell_2}}\right),
\end{displaymath}
which is just the maximal support CCRB. However, if the matrix perturbation is large enough so that $\sigma_e^2\gg 1/s$, the difference between the HCRB and CCRB cannot be neglected. Therefore $1/\sqrt{s}$ can be regarded as the transition point $\sigma_{e,\mathrm{t}}$ : when $\sigma_e<\sigma_{e,\mathrm{t}}$, the HCRB degenerates to the CCRB as $x_q\rightarrow+\infty$; when $\sigma_e>\sigma_{e,\mathrm{t}}$, the lower bound is raised for large $x_q$ and does not degenerate to the CCRB.

The theory of the transition point $\sigma_{e,\mathrm{t}}$ may have a geometrical explanation which is not rigorous but very intuitive. The set $\mathcal{X}_s$ can be regarded as the union of $s$-dimensional hyperplanes spanned by $s$ coordinate axes respectively. The sparse parameter $\mathbf{x}$ lies on one of the hyperplanes $\Sigma_s$ and is surrounded by a noise ball whose radius is $r^2=\sigma_\mathbf{x}^2$. As $x_q\rightarrow+\infty$, the radius satisfies $r^2\sim\sigma_e^2\|\mathbf{x}\|_{\ell_2}^2\sim s\sigma_e^2x_q$. The transition point corresponds to the tangency of the noise ball to one of the hyperplanes of $\mathcal{X}_s$ apart from $\Sigma_s$, and the high and low $\sigma_e$ regimes correspond to whether or not the noise ball intersects with another hyperplane (See Fig. \ref{fig_HCRB_noiseball}). It can be easily verified that this geometrical interpretation gives out correct value of the transition point $\sigma_{e,\mathrm{t}}$.

\begin{figure}
\centering
\includegraphics[width=.35\textwidth]{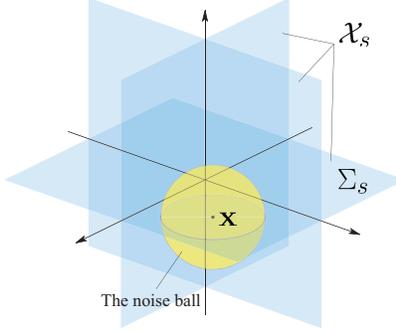}
\caption{The geometrical interpretation of the transition point $\sigma_{e,\mathrm{t}}$. The hyperplanes are part of the set $\mathcal{X}_s$, and the sparse parameter $\mathbf{x}$ lies on $\Sigma_s$. The radius of the noise ball is $r^2\sim s\sigma_e^2x_q$. If this ball intersects another hyperplane, such situation belongs to the high $\sigma_e$ regime; otherwise it belongs to the low $\sigma_e$ regime. For the low $\sigma_e$ regime, the HCRB equals the corresponding CCRB approximately if $x_q$ is sufficiently large; for the high $\sigma_e$ regime, the HCRB is evidently higher than the corresponding CCRB.}
\label{fig_HCRB_noiseball}
\end{figure}

\section{Numerical Results}\label{sec:num_results}

In this section, numerical simulations are performed in order to substantiate the theoretical results
presented in the previous sections.

\subsection{Numerical Analysis of $\gamma_\mathrm{CCRB}$}

Numerical experiments are first made on the CCRB, or equivalently, on the quantity $\gamma_\mathrm{CCRB}$. We wish to verify that the formula \eqref{approx_dratio} is valid and can demonstrate how $\gamma_\mathrm{CCRB}$ varies versus the perturbation level $c_e$ and the noise level $c_n$. Before we perform the numerical simulations, it is worthwhile to analyze the formula \eqref{approx_dratio} first with the help of its graph obtained by numerical approaches.

Define the function in the approximate formula \eqref{approx_dratio}
\begin{equation}
\gamma(c_e, c_n) = \frac{1}{s}\frac{2c_e}{2c_e+1+c_n/c_e}.
\end{equation}
The graphs of the function $\gamma(c_e, c_n)$ with varying $c_e$ and $c_n$ are shown in Fig. \ref{fig_CCRB1_cn_vary}, \blue{
where the sparsity $s$ is set to be $10$}
. It can be seen that when $c_n$ is fixed, $\gamma(c_e, c_n)$ is monotonically increasing of $c_e$, with limits $\gamma(0^+, c_n)=0$ and $\gamma(+\infty, c_n)=1/s$. Also, for each fixed $c_n$, there exists a transition point $c_{e,\mathrm{t}}$ of the curve which approximately separates the high $c_e$ and low $c_e$ regimes: when $c_e\ll c_{e,\mathrm{t}}$, one has $\gamma(c_e, c_n)\approx0$, while for $c_e\gg c_{e,\mathrm{t}}$, one has $\gamma(c_e, c_n)\approx1/s$. This transition point can be defined such that $\gamma(c_{e,\mathrm{t}},c_n)=1/2s$. When $c_n=0$, the transition point is $c_{e,\mathrm{t}}=1/2=-3~\mathrm{dB}$; for general $c_n>0$, $c_{e,\mathrm{t}}$ is the positive root of $2x^2-x-c_n=0$.

When $c_e$ is fixed, $\gamma(c_e, c_n)$ is monotonically decreasing of $c_n$. Fig. \ref{fig_CCRB1_cn_vary} also indicates that the transition point $c_{e,\mathrm{t}}(c_n)$ is a monotonically increasing function of $c_n$, which is a direct corollary of the monotonicity of $\gamma(c_e, c_n)$ with respect to $c_n$.

\begin{figure}
\centering
\includegraphics[width=.48\textwidth]{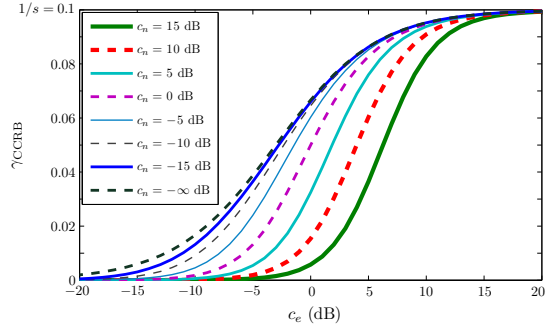}
\caption{The graphs of the approximate formula of $\gamma_\mathrm{CCRB}$ as a function of $c_e$ with different settings of $c_n$. The theoretical approximate formula is given by \eqref{approx_dratio}.}
\label{fig_CCRB1_cn_vary}
\end{figure}

In order to verify that the trend of $\gamma_\mathrm{CCRB}$ can be described by $\gamma(c_e, c_n)$, numerical results of $\gamma_\mathrm{CCRB}$ are performed. The dimensions are $n=20s,m=10s$, and $s=10$. The entries of $\mathbf{A}$ are drawn independently from $\mathcal{N}(0,1/m)$. Such generation of $\mathbf{A}$ is standard in the field of compressive sensing, and ensures the existence of the constants $\vartheta_{\mathrm{l},s}$ and $\vartheta_{\mathrm{u},s}$ with overwhelming probability \cite{Candes_DecodingLP,Blanchard_RIP}. The support of $\mathbf{x}$ is equiprobably chosen from all subsets of $\{1,\ldots,n\}$ with size $s$, and the nonzero entries satisfy i.i.d. equiprobable Bernoulli distribution on $\{-1,1\}$.

The results are shown in Fig. \ref{fig_CCRB1}. Points marked by ``$\times$'' are raw simulation results, while the solid lines represent the function $\gamma(c_e,c_n)$. It can be seen that the curves of the function $\gamma(c_e,c_n)$ can correctly describe how $\gamma_\mathrm{CCRB}$ varies versus $c_e$ and $c_n$. The figures also demonstrate that the solid curves lie almost in the middle of the raw points, which partially justifies the validity of the approximate formula and the bounds of \eqref{ubound_dCCRB}.

\begin{figure*}[!t]
\centerline{\subfigure[$c_n=-\infty$ dB]{\includegraphics[width
=.32\textwidth]{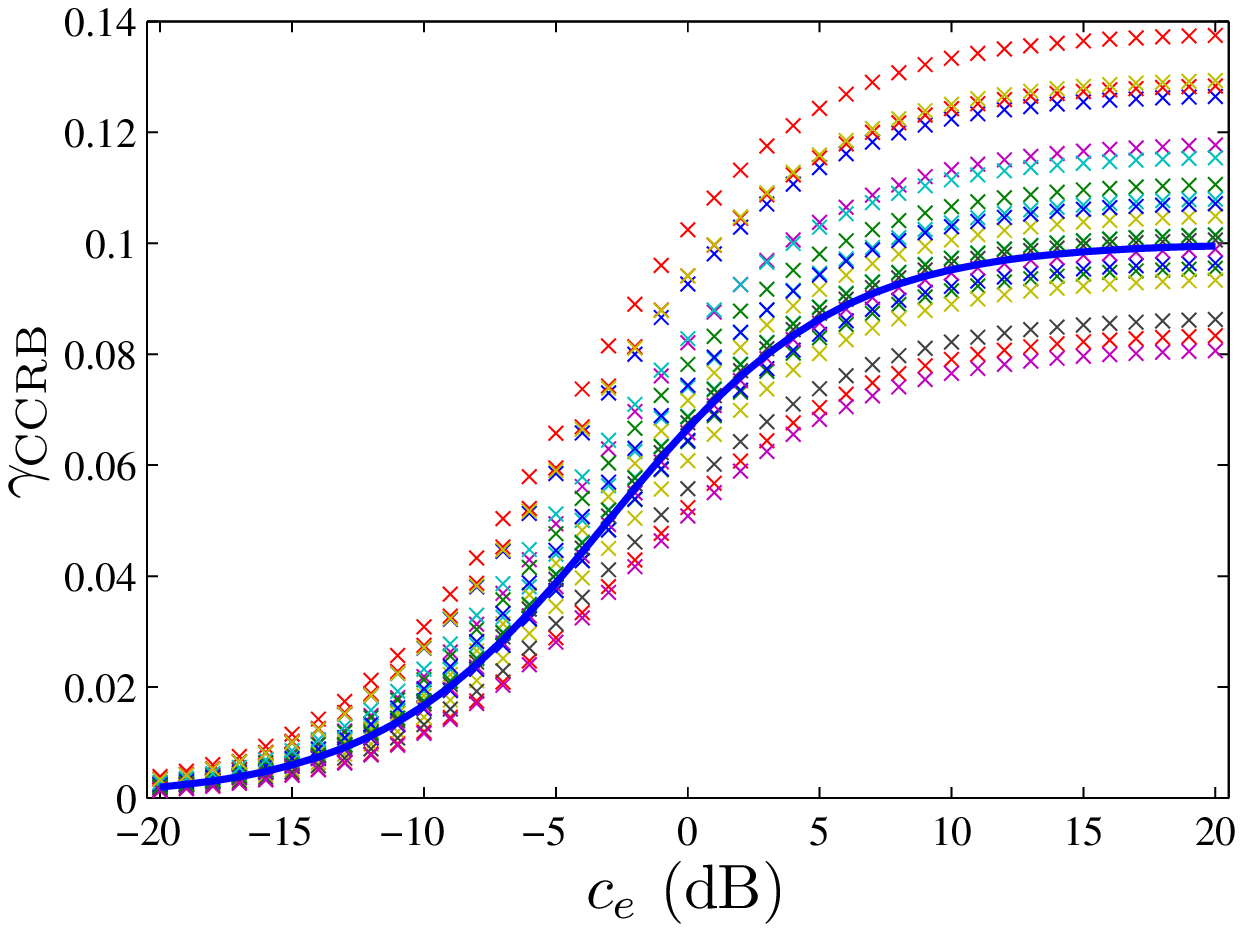}
\label{fig_CCRB1_cn0}}
\hfil
\subfigure[$c_n=-5$ dB]{\includegraphics[width=.32\textwidth]{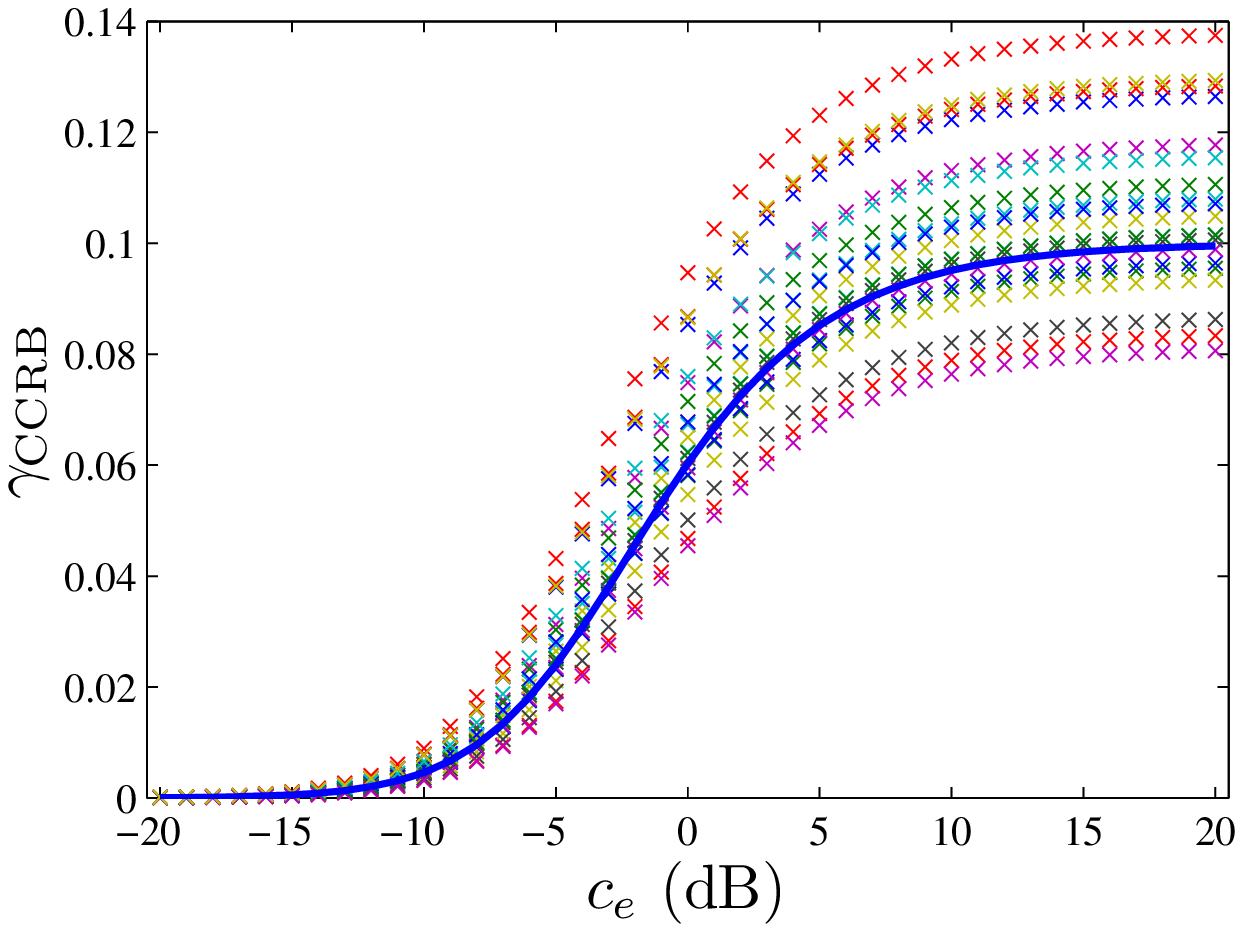}
\label{fig_CCRB1_cn-5dB}}
\hfil
\subfigure[$c_n=15$ dB]{\includegraphics[width=.32\textwidth]{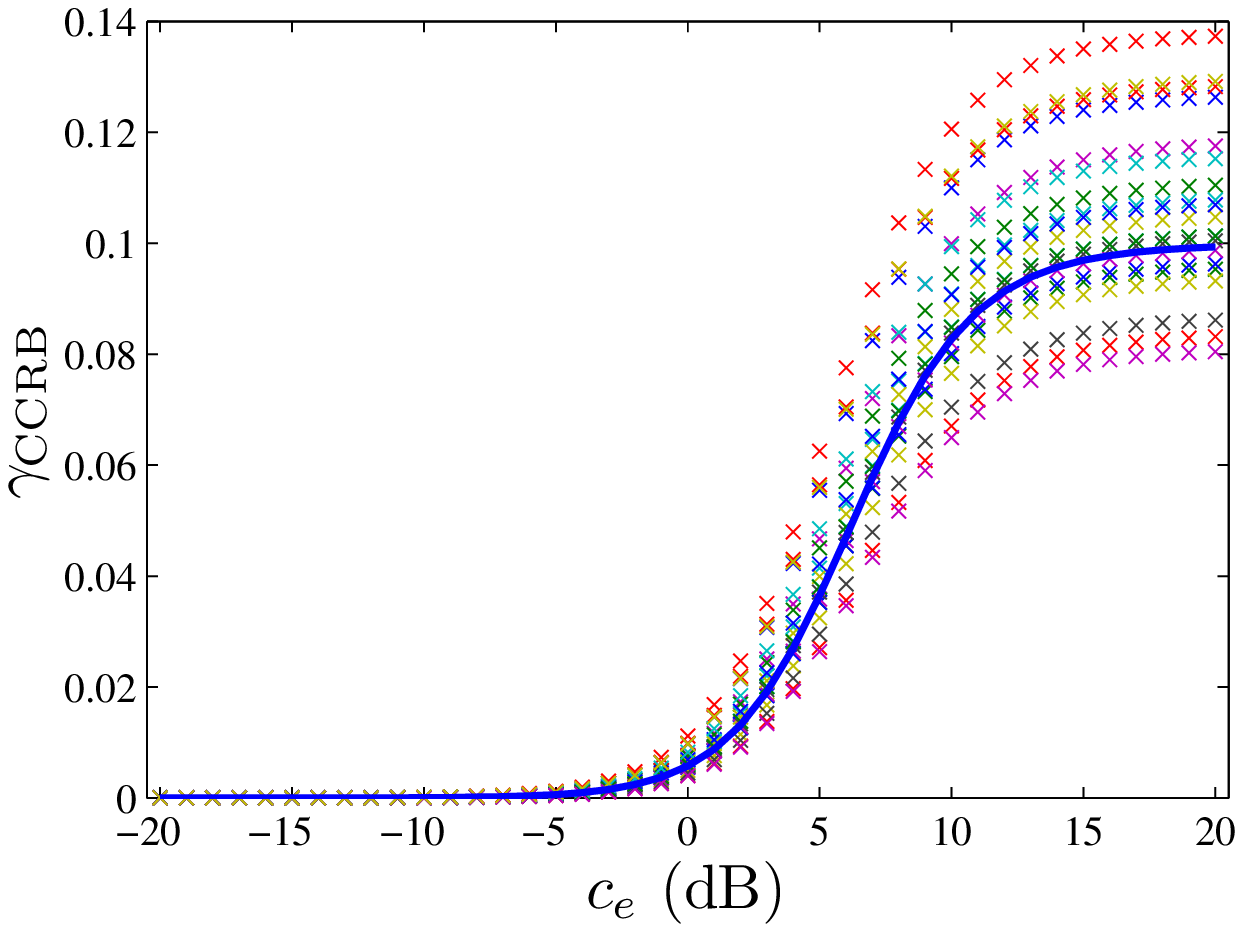}
\label{fig_CCRB1_cn15dB}}
}
\caption{$\gamma_\mathrm{CCRB}$ as a function of $c_e$ with different settings of $c_n$. Data points marked by ``$\times$'' are generated by simulation results. The solid lines represent the corresponding values calculated by the approximate expression \eqref{approx_dratio}.}
\label{fig_CCRB1}
\end{figure*}

Next we wish to verify the asymptotic behavior of $\gamma_\mathrm{CCRB}$, i.e. the $s^{-1}$ law given by Corollary \ref{corol:asymp}. This time
$s$ varies from $3$ to $300$ with exponentially increasing increments, and $n=20s, m=10s$. The generations of $\mathbf{A}$ and $\mathbf{x}$ are the same as in previous simulations.

The results are shown in Fig. \ref{fig_CCRB2}. The dotted points are raw experimental data of $\gamma_\mathrm{CCRB}$, and the solid straight line is the $s^{-1}$ line given by \eqref{approx_dratio}. Note that we employ the log-log scaling of the coordinate system so that the $s^{-1}$ law is represented by a straight line with slope $-1$. It can be seen that the data points can be upper and lower bounded by two $s^{-1}$ curves, and they cluster near the straight line given by \eqref{approx_dratio}. Thus it can be concluded that the trend of $\gamma_\mathrm{CCRB}$ with respect to $s$ can be described by the $s^{-1}$ law.

\begin{figure*}[!t]
\centerline{
\subfigure[$c_n=-\infty$ dB]{\includegraphics[width=.32\textwidth]{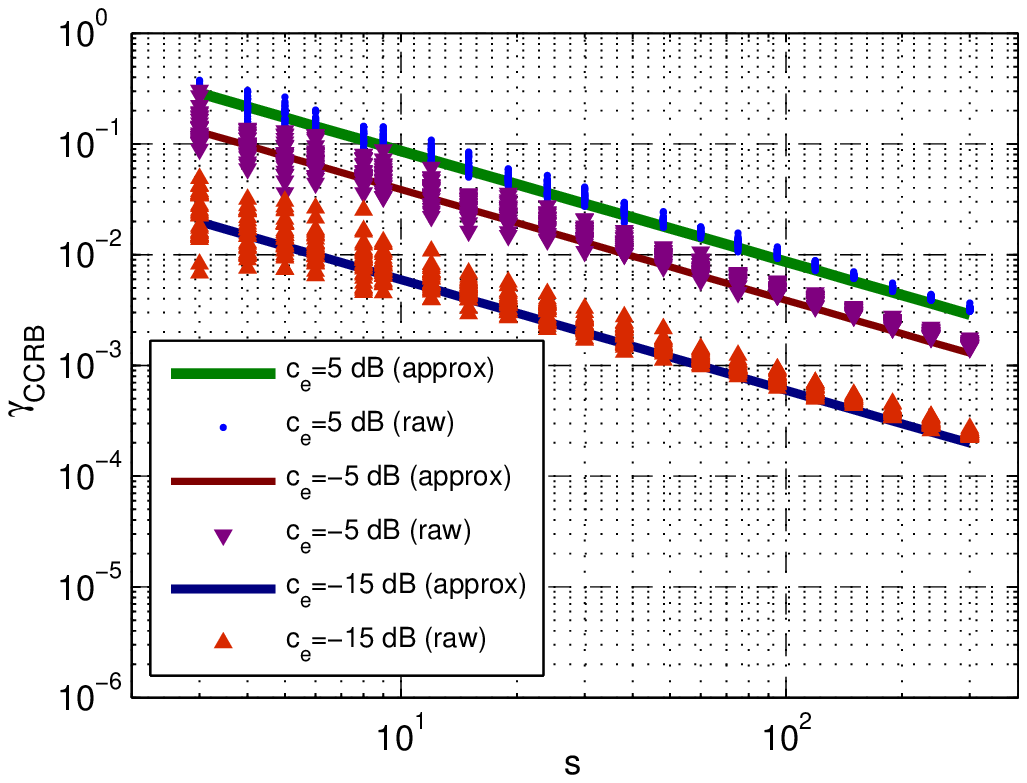}
\label{fig_CCRB2_cn0}}
\hfil
\subfigure[$c_n=-15$ dB]{\includegraphics[width
=.32\textwidth]{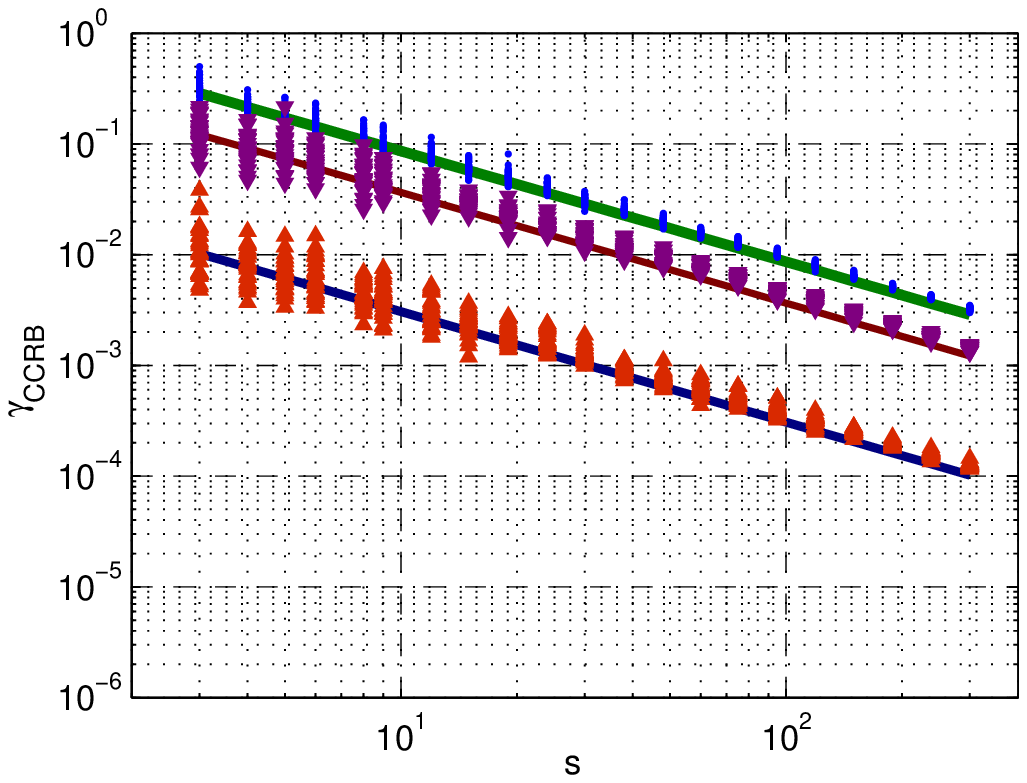}
\label{fig_CCRB2_cn-15dB}}
\hfil
\subfigure[$c_n=-5$ dB]{\includegraphics[width=.32\textwidth]{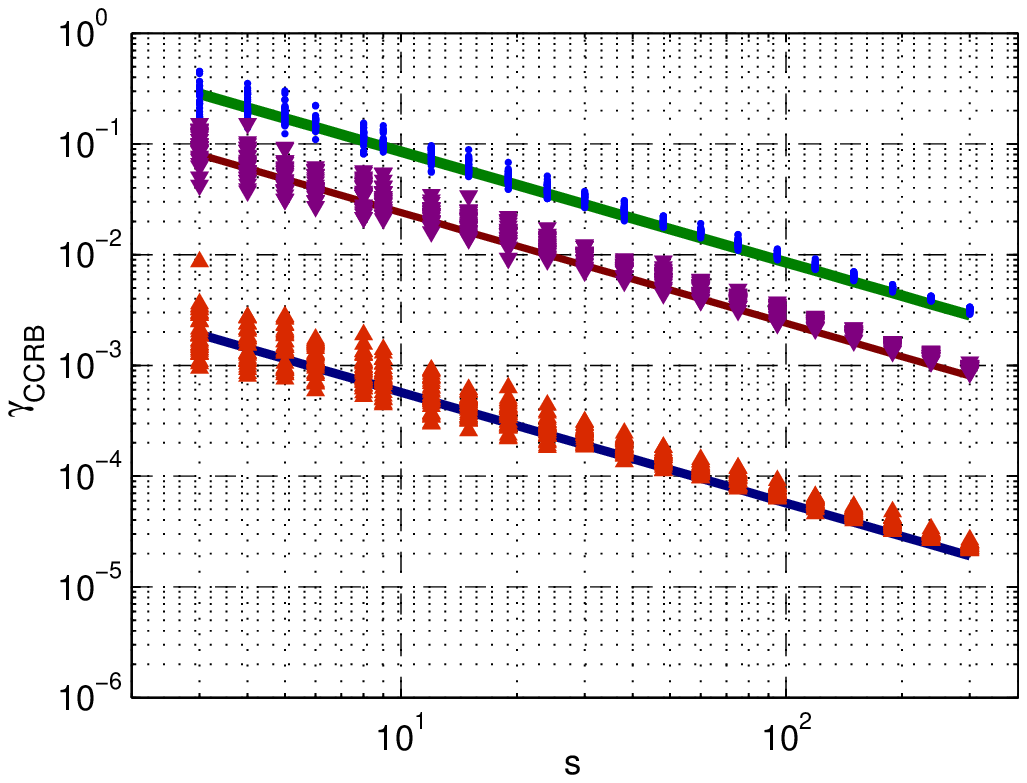}
\label{fig_CCRB2_cn-5dB}}
}
\caption{Simulation results of $\gamma_\mathrm{CCRB}$ versus the support size $s$ under different settings of $c_e$ and $c_n$. Dotted data points are generated by numerical simulations, and the solid straight lines are computed by the approximate expression \eqref{approx_dratio}. Every single sub-figure contains three groups of data, which correspond to $c_e=5$ dB,$c_e=-5$ dB and $c_e=-15$ dB as shown from the top to the bottom.
The figures share the same legend as shown in Fig.\ref{fig_CCRB2_cn0}. Notice that the log-log scaling of the coordinate system makes the graph of $s^{-1}$ a straight line of slope $-1$.}
\label{fig_CCRB2}
\end{figure*}

\subsection{Numerical Analysis of the HCRB}

We first have a brief review of Theorem \ref{thm:HCRB}. The lower bound of \eqref{HCRB_closedform} consists of two parts: the maximal support CCRB and the additional term which eliminates the CCRB gap. Numerical experiments will mainly be performed on the additional term, and the common factor $\sigma_\mathbf{x}^2$ will be ignored. In other words, our analysis will be focussed on the quantity
\begin{equation}
\begin{aligned}
d_\mathrm{HCRB}
=\frac{(n-s)\beta{\rm e}^{-\beta}}{{\rm e}^\beta-1}\left(1-\frac{1}{n-s+{\rm e}^\beta(1-g(\beta))^{-1}}\right).
\end{aligned}
\end{equation}

Settings of the numerical simulations are as follows. The dimensions are $n=m=10s$. $\mathbf{x}$ is set to be $[x_q,...,x_q,0,\ldots,0]^\mathrm{T}$. Different settings of $\sigma_e$ and $\sigma_n$ are employed, and for each group of $\{\sigma_e, \sigma_n\}$ we simulate the HCRB with varying $x_q$ to get a graph of $d_\mathrm{HCRB}$.

The first task is to verify the theory of transition point presented previously. The sparsities are $1,3,10,30$ and $100$ respectively, and for each $s$, $\sigma_e^2$ ranges from $0.0001$ to $100$. $\sigma_n$ is set to be $0.1$, and the HCRB for $x_q=1000$ is computed as a good approximation for the $x_q\rightarrow+\infty$ case.

The numerical results are shown in Fig. \ref{fig_HCRB_sigma_e}, where $d_\mathrm{HCRB}$ is ``normalized'' by $n-s$ so that the curves share a similar scale. It can be seen that the transition point for each $s$ exists and can be well evaluated by $\sigma_{e,\mathrm{t}}^2=1/s$. For $\sigma_e\ll\sigma_{e,\mathrm{t}}$, $d_\mathrm{HCRB}/(n-s)$ is much less than $1$, and as $\sigma_e\rightarrow0$, $d_\mathrm{HCRB}/(n-s)$ tends to zero rapidly; for $\sigma_e>\sigma_{e,\mathrm{t}}$, $d_\mathrm{HCRB}/(n-s)$ cannot be neglected, and as $\sigma_e\rightarrow+\infty$, $d_\mathrm{HCRB}/(n-s)$ possesses a limit which is of order $1$.

\begin{figure}[!t]
\centering
\includegraphics[width
=.48\textwidth]{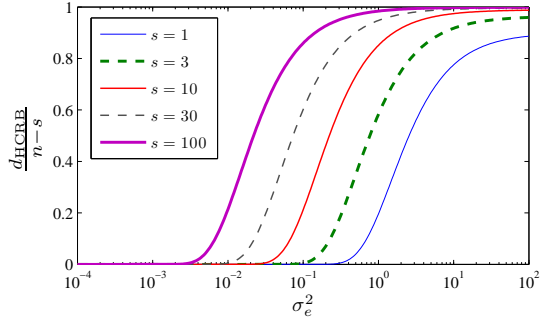}
\caption{$d_\mathrm{HCRB}/(n-s)$ versus $\sigma_e$ with different settings of $s$.}
\label{fig_HCRB_sigma_e}
\end{figure}

Next we fix $\sigma_e$ and $\sigma_n$, and test how $d_\mathrm{HCRB}$ varies versus $x_q$. The sparsity $s$ is set to be $1$ for simplicity. Results are shown in Fig. \ref{fig_HCRB1_sig_n0.1}, where $\sigma_n=0.1$ is fixed with different $\sigma_e$'s. The figure shows that when $\sigma_e$ belongs to the low value regime (i.e. $\sigma_e\ll\sigma_{e,\mathrm{t}}$, where $\sigma_{e,\mathrm{t}}=1$ here), its effect on $d_\mathrm{HCRB}$ can be neglected. Moreover, each curve can be separated into three regions: the low $x_q$ region where $\mathrm{HCRB}\approx\mathrm{CCRB}_\textrm{nonmax supp}$, the high $x_q$ region where $\mathrm{HCRB}\approx\mathrm{CCRB}_\textrm{max supp}$, and the transition region which connects the low and high $x_q$ region. However, when $\sigma_e$ exceeds $\sigma_{e,\mathrm{t}}$, the behavior of $d_\mathrm{HCRB}$ will become rather different and exhibits a more complex pattern. These phenomena demonstrate the way matrix perturbation influences the HCRB.

\begin{figure}[!t]
\centering
\includegraphics[width=.48\textwidth]{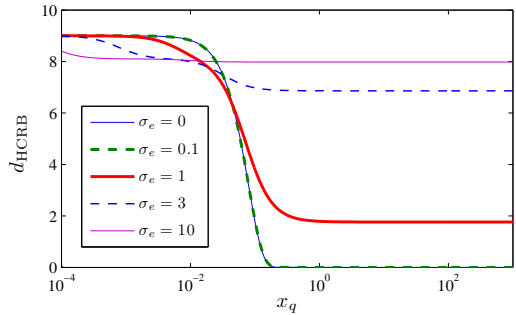}\label{fig_HCRB1_sig_n0.1}
\caption{$d_\mathrm{HCRB}$ versus $x_q$ with $\sigma_n=0.1$ and different settings of $\sigma_e$.}
\end{figure}

\blue{\subsection{Comparison with Existing Estimators}}

\blue{In this subsection we compare the derived bounds with the performance of existing estimators. Because the HCRB is only derived for the unit
sensing matrix case, the following analysis and simulations will be focused on this particular situation. One of the existing estimators to be compared
is the maximum likelihood (ML) estimator, which, in the unit sensing matrix case, is given by
\begin{equation}
\hat{\mathbf{x}}_{\mathrm{ML}}(\mathbf{y}) = \mathcal{P}_s(\mathbf{y}),
\end{equation}
where the operator $\mathcal{P}_s$ retains the $s$ largest entries in magnitude and zeros out all others. It should be noted that when
the noise is large, this estimator will be severely biased. Another estimator is the one given in \cite[\S 5]{Eldar_SLM_RKHS} for the special case of $s=1$. This estimator is globally unbiased, and is given by
\begin{equation}\label{eq_unbiased_estim}
\hat{x}_k(\mathbf{y})
=\left\{
\begin{aligned}
&y_q,&\quad&k=q, \\
&y_k\exp\left(-\frac{2y_k x_{0,q}+x_{0,q}^2}{2\sigma_{\mathbf{x}_0}^2}\right),
&\quad&\textrm{otherwise}.
\end{aligned}
\right.
\end{equation}
Here $\mathbf{x}_0$ is a sparse vector with $s=1$ which represents
a kind of prior knowledge that the true parameter $\mathbf{x}$ will not be far from $\mathbf{x}_0$\footnote{\blue{
Rigorously, this estimator is the locally minimum variance unbiased (LMVU) estimator at $\mathbf{x}_0$ when $\sigma_n=0$. See \cite{Eldar_SLM_RKHS} for more detailed explanations.}}, $q$ is the index of the single non-zero entry, and $x_{0,q}$ is the value of this entry.}

\blue{\begin{figure}[!t]
\centering
\subfigure[$\sigma_e=0.1$]{\includegraphics[width=.48\textwidth]{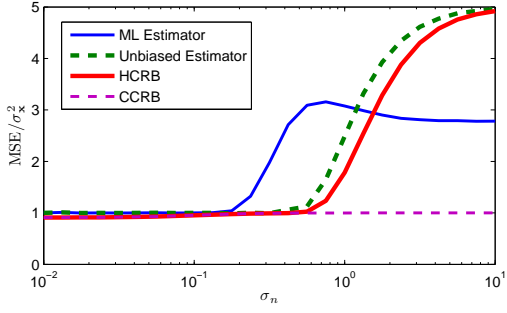}
\label{fig_Estim_1}}
\\
\subfigure[$\sigma_e=1$]{\includegraphics[width=.48\textwidth]{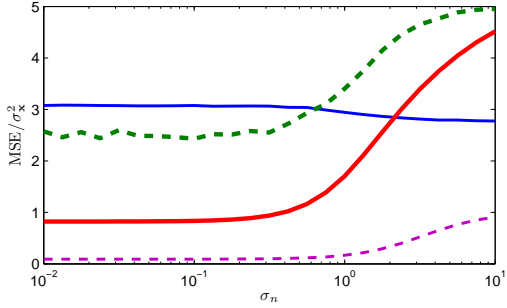}
\label{fig_Estim_2}}
\caption{Comparison of the performance of ML estimator and the unbiased estimator \eqref{eq_unbiased_estim} with HCRB and CCRB.}
\label{fig_Estim_Perform}
\end{figure}}

\blue{We test the case when $s=1$ and $n=5$. The sparse vector $\mathbf{x}$ is set to be
$[1,0,\ldots,0]^\mathrm{T}$. The MSE of the two estimators together with HCRB and CCRB are numerically
computed for different settings of $\sigma_n$'s and $\sigma_e's$.}

\blue{The results are shown in Fig \ref{fig_Estim_Perform}. We first analyze the case $\sigma_e=0.1$ which represents relatively small
matrix perturbation. It can be seen that when $\sigma_n$ is small, the MSEs of the ML estimator and the unbiased estimator are
both lower bounded by the HCRB and the CCRB. As $\sigma_n$ increases, the ML estimator undergoes an evident performance degradation
before the unbiased estimator does. However, when $\sigma_n$ is considerably large, the MSE of the ML estimator is not lower bounded
by the HCRB. This is due to the severe bias of the ML estimator in such situations. For the unbiased estimator, it can be seen that
its MSE is always lower bounded by the HCRB.}

\blue{For the case $\sigma_e=1$, it can be seen that when $\sigma_n$ is small, the MSEs of both the ML estimator and the unbiased estimator are much larger than the HCRB due to the large matrix perturbation level. As explained before, the ML estimator is severely biased when the noise level is large, and thus for certain values of $\sigma_n$, its MSE will not be lower bounded by the HCRB. For the unbiased estimator, its MSE is always lower bounded by the HCRB.}

\section{Conclusion}\label{sec:conclusion}

In this paper, the performance bound of sparse estimation with general perturbation has been studied. Two widely-used types of lower bounds, the CCRB and the HCRB, have been calculated and analyzed. For the CCRB, it has been shown that the additional term of the CCRB can be approximated by a simple $s^{-1}$ law if the sensing matrix satisfies RIP-type conditions. For the HCRB where only the case of unit sensing matrix is studied, it has been shown that the HCRB provides a tighter lower bound than the CCRB, and that it exhibits a satisfying transition behavior when the smallest entry of the parameter tends to zero. A geometrical interpretation of the theory of the HCRB has also been given.

There are several future directions to be explored. First, the HCRB is obtained only for the unit sensing matrix case, which is useful for
a qualitative comprehension but has apparent limitation. In many cases the sensing matrix cannot be assumed unit, and
a more precise quantitative study on the performance bound also requires generalizing the HCRB to the general sensing matrix case.
There may be two practical approaches to this problem. One is to derive a closed-form lower bound which is convenient to analyze and
understand; the other is to find a tractable way to numerically compute the lower bound. However, both ways need further research and
are still waiting for useful results.

Second, the HCRB provides a lower bound for globally unbiased estimators. However, recovery algorithms are quite likely to be biased
in the sparse setting, and the bias is usually dependent on the noise variance. Moreover, there are occasions that biased estimators can achieve a lower MSE than unbiased estimators \cite{Eldar_Coherence}. These problems also point out a possible direction for future study.

\appendices

\section{Proof of Lemma \ref{lemma:FIM}}\label{proof_FIM}

We first compute the likelihood function $p(\mathbf{y};\mathbf{x})$. Notice that the problem \eqref{fund_eq} can be re-formulated as
\begin{equation}\label{fund_eq_form2}
\mathbf{y}=\mathbf{A}\mathbf{x}+\mathbf{E}\mathbf{x}+\mathbf{n}=\mathbf{A}\mathbf{x}+\mathbf{n}_\mathbf{x},
\end{equation}
where $\mathbf{n}_\mathbf{x}=\mathbf{E}\mathbf{x}+\mathbf{n}$ denotes the equivalent noise. Because the elements of $\mathbf{E}$ are drawn i.i.d. from Gaussian distribution, $\mathbf{E}\mathbf{x}$ is an $m$-dimensional Gaussian-distributed random vector, and is independent of $\mathbf{n}$.
By straightforward calculation, it can be shown that $E[\mathbf{E}\mathbf{x}]=0$ and
$\mathrm{Cov}(\mathbf{E}\mathbf{x}) = \sigma_e^2\|\mathbf{x}\|^2_{\ell_2}\mathbf{U}_m$.
Then, from the mutual independence of $\mathbf{E}\mathbf{x}$ and $\mathbf{n}$ one has $\mathbf{n}_\mathbf{x}\sim \mathcal{N}(0,(\sigma_e^2\|\mathbf{x}\|_{\ell_2}^2+\sigma_n^2)\mathbf{U}_m)$. Therefore the likelihood function is given by
\begin{equation}\label{likelihood}
\begin{aligned}
L(\mathbf{x})=\frac{1}{\left(2\pi\sigma_{\mathbf{x}}^2\right)^{\frac{m}{2}}}
\exp\left[-\frac{(\mathbf{y}-\mathbf{A}\mathbf{x})^\mathrm{T}(\mathbf{y}-\mathbf{A}\mathbf{x})}{2\sigma_{\mathbf{x}}^2}\right],
\end{aligned}
\end{equation}
where
\begin{equation}
\sigma_{\mathbf{x}}^2=\sigma_e^2\|\mathbf{x}\|_{\ell_2}^2+\sigma_n^2.
\end{equation}

The next step is to compute the FIM. The gradient of the log likelihood function $\ln p(\mathbf{y};\mathbf{x})$ is
\begin{equation}
\begin{aligned}
\nabla_{\mathbf{x}}\ln p(\mathbf{y};\mathbf{x})
=\frac{1}{\sigma_e^2\mathbf{x}^\mathrm{T}\mathbf{x}+\sigma_n^2}\Bigg[\mathbf{A}^\mathrm{T}(\mathbf{y}-\mathbf{A}\mathbf{x})+\frac{\|\mathbf{y}-\mathbf{A}\mathbf{x}\|_{\ell_2}^2\mathbf{x}}{\sigma_e^2\mathbf{x}^\mathrm{T}\mathbf{x}+\sigma_n^2}\Bigg]
-m\frac{\sigma_e^2\mathbf{x}}{\sigma_e^2\mathbf{x}^\mathrm{T}\mathbf{x}+\sigma_n^2},
\end{aligned}
\end{equation}
and therefore
\begin{equation}\label{FIM_intermediate}
\begin{aligned}
(\nabla_{\mathbf{x}}\ln p(\mathbf{y};\mathbf{x}))(\nabla_{\mathbf{x}}^\mathrm{T}\ln p(\mathbf{y};\mathbf{x}))
=~&\frac{m^2\sigma_e^4}{(\sigma_e^2\mathbf{x}^\mathrm{T}\mathbf{x}+\sigma_n^2)^2}\mathbf{x}\mathbf{x}^\mathrm{T}
-\frac{2m\sigma_e^4\|\mathbf{y}-\mathbf{A}\mathbf{x}\|_{\ell_2}^2}{(\sigma_e^2\mathbf{x}^\mathrm{T}\mathbf{x}+\sigma_n^2)^3}\mathbf{x}\mathbf{x}^\mathrm{T} \\
&+\frac{1}{(\sigma_e^2\mathbf{x}^\mathrm{T}\mathbf{x}+\sigma_n^2)^2}\mathbf{A}^\mathrm{T}(\mathbf{y}-\mathbf{A}\mathbf{x})(\mathbf{y}-\mathbf{A}\mathbf{x})^\mathrm{T}
\mathbf{A}\\
&+\frac{\sigma_e^4\|\mathbf{y}-\mathbf{A}\mathbf{x}\|_{\ell_2}^4}{(\sigma_e^2\mathbf{x}^\mathrm{T}\mathbf{x}+\sigma_n^2)^4}\mathbf{x}\mathbf{x}^\mathrm{T} +\textrm{terms linear in }(\mathbf{y}-\mathbf{A}\mathbf{x}) \\
&+\textrm{terms cubic in }(\mathbf{y}-\mathbf{A}\mathbf{x}).
\end{aligned}
\end{equation}
The specific forms of the linear and cubic terms in the above equation is of no importance; they will vanish when we take their expectation because $\mathbf{y}-\mathbf{A}\mathbf{x}$ is Gaussian-distributed with zero mean and diagonal covariance. The expectation value of $\|\mathbf{y}-\mathbf{A}\mathbf{x}\|_{\ell_2}^2$ is $m(\sigma_e^2\mathbf{x}^\mathrm{T}\mathbf{x}+\sigma_n^2)$, while for $\|\mathbf{y}-\mathbf{A}\mathbf{x}\|_{\ell_2}^4$ the expectation will be $(m^2+2m)(\sigma_e^2\mathbf{x}^\mathrm{T}\mathbf{x}+\sigma_n^2)^2$, which can be seen from the fact that for a series of i.i.d. zero-mean Gaussian random variables $w_1,\ldots,w_k$ with the same variance $q^2$, one has
$E[(w_1^2+\cdots+w_k^2)^2]=(k^2+2k)q^4$. By taking the expectation of \eqref{FIM_intermediate}, and recalling the covariance matrix of $\mathbf{y}-\mathbf{A}\mathbf{x}=\mathbf{E}\mathbf{x}+\mathbf{n}$, one has
\begin{equation*}
\begin{aligned}
\mathbf{J}(\mathbf{x})=\frac{1}{\sigma_e^2\mathbf{x}^\mathrm{T}\mathbf{x}+\sigma_n^2}\mathbf{A}^\mathrm{T}\mathbf{A}
+\frac{2m\sigma_e^4}{(\sigma_e^2\mathbf{x}^\mathrm{T}\mathbf{x}+\sigma_n^2)^2}\mathbf{x}\mathbf{x}^\mathrm{T},
\end{aligned}
\end{equation*}
which is exactly \eqref{FIM_xS}.

\section{Proof of Theorem \ref{thm:bounds_dCCRB}}\label{proof_bounds_dCCRB}

Assumption \ref{assump:sensing_mat} implies that the eigenvalues of $\mathbf{A}_S^\mathrm{T}\mathbf{A}_S$ are bounded by
$1-\vartheta_{\mathrm{l},s}$ and $1+\vartheta_{\mathrm{u},s}$.
Therefore the eigenvalues of $(\mathbf{A}_S^\mathrm{T}\mathbf{A}_S)^{-1}$, denoted by $\tilde{\lambda}_i$, are bounded as follows
\begin{equation}
\frac{1}{1+\vartheta_{\mathrm{u},s}}\leq \tilde{\lambda}_i \leq \frac{1}{1-\vartheta_{\mathrm{l},s}}.
\end{equation}
\blue{
Since $(\mathbf{A}_S^\mathrm{T}\mathbf{A}_S)^{-1}$ is a symmetric matrix, it can be diagonalized by some orthogonal matrix $\mathbf{Q}$, i.e.
\begin{equation*}
(\mathbf{A}_S^\mathrm{T}\mathbf{A}_S)^{-1}=\mathbf{Q}^\mathrm{T}\mathrm{diag}(\tilde{\lambda}_1,\ldots,\tilde{\lambda}_s) \mathbf{Q}.
\end{equation*}
Thus, for any vector $\mathbf{u}\in\mathbb{R}^s$, denote $\tilde{\mathbf{u}}=\mathbf{Q}\mathbf{u}$, one has
\begin{equation*}
\begin{aligned}
\mathbf{u}^\mathrm{T}(\mathbf{A}_S^\mathrm{T}\mathbf{A}_S)^{-1}\mathbf{u}
=&\tilde{\mathbf{u}}^\mathrm{T}\mathrm{diag}(\tilde{\lambda}_1,\ldots,\tilde{\lambda}_s)\tilde{\mathbf{u}} =\sum_{k=1}^s\tilde{\lambda}_k\tilde{u}_k^2\\
\geq&\sum_{k=1}^s\frac{1}{1+\vartheta_{\mathrm{u},s}}\tilde{u}_k^2
=\frac{1}{1+\vartheta_{\mathrm{u},s}}\|\tilde{\mathbf{u}}\|_{\ell_2}^2\\
=&\frac{1}{1+\vartheta_{\mathrm{u},s}}\|\mathbf{u}\|_{\ell_2}^2.
\end{aligned}
\end{equation*}
By employing the same techniques it can be verified that
\begin{equation}\label{norm_ineq1}
\frac{\|\mathbf{x}\|_{\ell_2}^2}{1+\vartheta_{\mathrm{u},s}}\leq\mathbf{x}_S^\mathrm{T}(\mathbf{A}_S^\mathrm{T}\mathbf{A}_S)^{-1}\mathbf{x}_S
\leq\frac{\|\mathbf{x}\|_{\ell_2}^2}{1-\vartheta_{\mathrm{l},s}},
\end{equation}
where we have used the fact that for $s$-sparse vector $\mathbf{x}$ supported on $S$, one has $\|\mathbf{x}_S\|_{\ell_2}^2=\|\mathbf{x}\|_{\ell_2}^2$. Noting that
the square of $(\mathbf{A}_S^\mathrm{T}\mathbf{A}_S)^{-1}$ has eigenvalues
$\tilde{\lambda}_1^2,\ldots,\tilde{\lambda}_s^2$ with the same eigenvectors as
$(\mathbf{A}_S^\mathrm{T}\mathbf{A}_S)^{-1}$, it can be shown similarly that
\begin{equation}\label{norm_ineq2}
\frac{\|\mathbf{x}\|_{\ell_2}^2}{(1+\vartheta_{\mathrm{u},s})^2}\leq\|(\mathbf{A}_S^\mathrm{T}\mathbf{A}_S)^{-1}\mathbf{x}_S\|^2_{\ell_2}\leq \frac{\|\mathbf{x}\|_{\ell_2}^2}{(1-\vartheta_{\mathrm{l},s})^2}.
\end{equation}}
Substituting \eqref{norm_ineq1} and \eqref{norm_ineq2} into \eqref{dCCRB}, one will obtain
\begin{equation}\label{temp_bound_dCCRB}
\begin{aligned}
d_\mathrm{CCRB}\lesseqqgtr \sigma_\mathbf{x}^2\frac{1+\vartheta_{\pm,s}}{(1+\vartheta_{\mp,s})^2} 
\cdot\frac{2m\sigma_e^2}{1+2m\sigma_e^2+\vartheta_{\pm,s}+\frac{c_n(1+\vartheta_{\pm,s})}{m\sigma_e^2}},\\
\end{aligned}
\end{equation}
where $\vartheta_{+,s}=\vartheta_{\mathrm{u},s},\vartheta_{-,s}=-\vartheta_{\mathrm{l},s}$.

Next we wish to bound the term $m\sigma_e^2$. It can be seen from the definition of $c_e$ that $m\sigma_e^2=c_e\tr(\mathbf{A}_S^\mathrm{T}\mathbf{A}_S)/s$, and
with the help of the bounds of the eigenvalues of $\mathbf{A}_S^\mathrm{T}\mathbf{A}_S$, it can be shown that
\begin{equation}
c_e(1-\vartheta_{\mathrm{l},s})\leq m\sigma_e^2\leq c_e(1+\vartheta_{\mathrm{u},s}).
\end{equation}
Substituting this into \eqref{temp_bound_dCCRB}, one will obtain the bounds given by \eqref{ubound_dCCRB}.

The bounds of $\gamma_\mathrm{CCRB}$ can also be derived by similar techniques with the help of the bounds of the eigenvalues of
$(\mathbf{A}_S^\mathrm{T}\mathbf{A}_S)^{-1}$.

\section{Proof of Lemma \ref{lem:HCRB_general}}\label{proof_HCRB_general}

We first refer to \cite{Gorman_HCRB} for the definition of the multivariate HCRB for unbiased estimators.

\begin{prop}\label{prop:HCRB_original}\cite{Gorman_HCRB}
Suppose $p(\mathbf{y};\mathbf{x})$ are a class of pdf's parameterized by $\mathbf{x}\in\mathcal{X}$, and let $\mathbf{x},\mathbf{x}+\mathbf{v}_1,\ldots,\mathbf{x}+\mathbf{v}_k$ be test points contained in the constrained parameter set $\mathcal{X}$. Define
\begin{align}
\mathbf{V}&=[\mathbf{v}_1,\ldots,\mathbf{v}_k], \label{HCRB_Vmat}\\
\delta_i p &= p(\mathbf{y};\mathbf{x}+\mathbf{v}_i)-p(\mathbf{y};\mathbf{x}), \\
\bs{\delta} p&=[\delta_1 p,\ldots,\delta_k p]^\mathrm{T}, \\
\mathbf{H} &= E_{\mathbf{y};\mathbf{x}}\left[
\frac{\bs{\delta}p}{p}\frac{\bs{\delta}p^\mathrm{T}}{p}\right], \label{H_mat_HCRB_general}
\end{align}
where $p$ denotes $p(\mathbf{y};\mathbf{x})$ for short.
Then for any unbiased estimator $\hat{\mathbf{x}}$, the estimator covariance matrix $\mathrm{Cov}(\hat{\mathbf{x}})$ satisfies the matrix inequality
\begin{equation}\label{HCRB_general}
\mathrm{Cov}(\hat{\mathbf{x}})
\succeq
\mathbf{V}\mathbf{H}^\dagger
\mathbf{V}^\mathrm{T},
\end{equation}
\end{prop}

Next Proposition \ref{prop:HCRB_original} will be applied to the present case. The $(i,j)$th element of $\mathbf{H}$ is given by
\begin{equation}\label{Hij_general_deriv}
\begin{aligned}
H_{ij}
=E_{\mathbf{y};\mathbf{x}}\left[\frac{p(\mathbf{y};\mathbf{x}+\mathbf{v}_i)p(\mathbf{y};\mathbf{x}+\mathbf{v}_j)}{p^2(\mathbf{y};\mathbf{x})}\right]-E_{\mathbf{y};\mathbf{x}}\left[\frac{p(\mathbf{y};\mathbf{x}+\mathbf{v}_i)}{p(\mathbf{y};\mathbf{x})}\right]
 - E_{\mathbf{y};\mathbf{x}}\left[\frac{p(\mathbf{y};\mathbf{x}+\mathbf{v}_j)}{p(\mathbf{y};\mathbf{x})}\right] + 1.
\end{aligned}
\end{equation}
The second term of the above equation is
\begin{equation}
-E_{\mathbf{y};\mathbf{x}}\left[\frac{p(\mathbf{y};\mathbf{x}+\mathbf{v}_i)}{p(\mathbf{y};\mathbf{x})}\right]
=-\int_{\mathbb{R}^m} p(\mathbf{y};\mathbf{x}+\mathbf{v}_i)\ud\mathbf{y}=-1,
\end{equation}
and similarly the third term also equals $-1$. The first term is
\begin{equation}
\begin{aligned}
&E_{\mathbf{y};\mathbf{x}}\left[\frac{p(\mathbf{y};\mathbf{x}+\mathbf{v}_i)p(\mathbf{y};\mathbf{x}+\mathbf{v}_j)}{p^2(\mathbf{y};\mathbf{x})}\right] \\
=~&\left(\frac{\sigma_\mathbf{x}^2}{2\pi\sigma_{\mathbf{x}+\mathbf{v}_i}^2\sigma_{\mathbf{x}+\mathbf{v}_j}^2}\right)^{\frac{m}{2}}
\exp\left(-\frac{\|\mathbf{A}\mathbf{v}_i\|_{\ell_2}^2}{2\sigma_{\mathbf{x}+\mathbf{v}_i}^2}-\frac{\|\mathbf{A}\mathbf{v}_j\|_{\ell_2}^2}{2\sigma_{\mathbf{x}+\mathbf{v}_i}^2}\right) \\
&\cdot
\int_{\mathbb{R}^m}
\exp\Bigg[\left(\frac{\mathbf{v}_i^\mathrm{T}\mathbf{A}^\mathrm{T}}{\sigma_{\mathbf{x}+\mathbf{v}_i}^2}
+\frac{\mathbf{v}_j^\mathrm{T}\mathbf{A}^\mathrm{T}}{\sigma_{\mathbf{x}+\mathbf{v}_j}^2}\right)(\mathbf{y}-\mathbf{A}\mathbf{x}) -\frac{\|\mathbf{y}-\mathbf{A}\mathbf{x}\|_{\ell_2}^2}{2\varsigma_{\mathbf{x},\mathbf{v}_i,\mathbf{v}_j}^2}\Bigg]\ud\mathbf{y} \\
=~&\left(\frac{\sigma_\mathbf{x}^2\varsigma_{\mathbf{x},\mathbf{v}_i,\mathbf{v}_j}^2}{\sigma_{\mathbf{x}+\mathbf{v}_i}^2\sigma_{\mathbf{x}+\mathbf{v}_j}^2}\right)^{\frac{m}{2}}
\exp\Bigg[-\frac{\|\mathbf{A}\mathbf{v}_i\|_{\ell_2}^2}{2\sigma_{\mathbf{x}+\mathbf{v}_i}^2}-\frac{\|\mathbf{A}\mathbf{v}_j\|_{\ell_2}^2}{2\sigma_{\mathbf{x}+\mathbf{v}_i}^2} +\frac{\varsigma_{\mathbf{x},\mathbf{v}_i,\mathbf{v}_j}^2}{2}\left\|\frac{\mathbf{A}\mathbf{v}_i}{\sigma_{\mathbf{x}+\mathbf{v}_i}^2}
+\frac{\mathbf{A}\mathbf{v}_j}{\sigma_{\mathbf{x}+\mathbf{v}_j}^2}\right\|_{\ell_2}^2
\Bigg],
\end{aligned}
\end{equation}
where
\begin{equation}
\frac{1}{\varsigma_{\mathbf{x},\mathbf{v}_i,\mathbf{v}_j}^2}
=\frac{1}{\sigma_{\mathbf{x}+\mathbf{v}_i}^2}+\frac{1}{\sigma_{\mathbf{x}+\mathbf{v}_j}^2}-\frac{1}{\sigma_{\mathbf{x}}^2}.
\end{equation}
Substituting these results into \eqref{Hij_general_deriv}, one will readily obtain \eqref{Hij_general}, and Lemma \ref{lem:HCRB_general} is proved.

\section{Proof of Theorem \ref{thm:HCRB}} \label{proof_HCRB_thm}

It is assumed without loss of generality that the support of $\mathbf{x}$ is $S=\{1,2,\ldots,s\}$. We split the MSE of an estimator $\hat{\mathbf{x}}$ into two parts as follows:
\begin{equation}
\begin{aligned}
\mathrm{mse}(\hat{\mathbf{x}})
=\sum_{i\in S}E[(\hat{x}_i-x_i)^2]+\sum_{i\notin S}E[(\hat{x}_i-x_i)^2],
\end{aligned}
\end{equation}
i.e. the support part and the non-support part. We choose different sets of $\mathbf{v}_i$ for the two parts respectively, and in the end combine these two parts to get a lower bound. This approach results from the fact that the MSE of a vector-valued estimator is the sum of the MSE of its components.

For the lower bound of the support part, the following $\{\mathbf{v}_i\}_{i=1}^{s}$ is employed:
\begin{equation}
\mathbf{v}_i=t\mathbf{e}_i,\quad i=1,\ldots,s,
\end{equation}
where $t$ is an arbitrary real number. After the corresponding covariance matrix is obtained, we take $t\rightarrow0$ and sum only the first $s$ diagonal elements to obtain the lower bound of the support part. It can be proved (see Appendix \ref{proof_CCRB_HCRB}) that this lower bound is identical to the CCRB.

For the non-support part, the following $\{\mathbf{v}_i\}_{i=1}^{n-s+1}$ is used:
\begin{equation}
\mathbf{v}_i=\left\{
\begin{aligned}
&x_q\mathbf{e}_{i+s}-x_q\mathbf{e}_q,&\quad&i=1,\ldots,n-s, \\
&t\mathbf{e}_q,&\quad& i=n-s+1.
\end{aligned}\right.
\end{equation}
Here $t$ is also an arbitrary real number which will tend to zero afterwards. At last, the last $n-s$ diagonal elements of the covariance matrix will be summed to represent the lower bound of the non-support part.

The matrix $\mathbf{V}$ in \eqref{HCRB_sparse_general} could be expressed as
\begin{equation}
\mathbf{V}
=\begin{bmatrix}
   t & -x_q\mathbf{1}^\mathrm{T} \\
   \mathbf{0} & x_q\mathbf{U}_{n-s} \\
 \end{bmatrix},
\end{equation}
(the order of $\mathbf{v}_i$'s is slightly changed which has no effect on the final result), where the bold face $\mathbf{1}$ denotes the column vector $[1,1,\ldots,1]^\mathrm{T}$. Then the elements of the matrix $\mathbf{H}$ can be given as follows:
\begin{equation}
H_{11}=\left(\frac{\sigma_\mathbf{x}^2\varsigma_{\mathbf{x},q}^2}{\sigma_q^4}\right)^{\frac{n}{2}}
\exp\left(\frac{t^2}{\sigma_q^2}\left(2\frac{\varsigma_{\mathbf{x},q}^2}{\sigma_q^2}-1\right)\right)-1,
\end{equation}
where it has been defined that
\begin{equation}
\begin{aligned}
\sigma_q^2&=\sigma_n^2+\sigma_e^2(\|\mathbf{x}+t\mathbf{e}_q\|_{\ell_2}^2), \quad
\varsigma_{\mathbf{x},q}^2=\left(\frac{2}{\sigma_q^2}-\frac{1}{\sigma_\mathbf{x}^2}\right)^{-1},
\end{aligned}
\end{equation}
and
\begin{align}
H_{1i}=H_{i1}&=\exp\left(-\frac{t x_q}{\sigma_\mathbf{x}^2}+\frac{x_q^2}{\sigma_\mathbf{x}^2}\left(\frac{\sigma_q^2}{\sigma_\mathbf{x}^2}-1\right)\right)-1,\\
H_{ij} &= \exp\left(\frac{(1+\delta_{ij})x_q^2}{\sigma_\mathbf{x}^2}\right)-1\label{Hij_proof_HCRB_thm}
\end{align}
for $i,j\geq 2$. The matrix $\mathbf{H}$ could be represented as
\begin{equation}
\mathbf{H}=
\begin{bmatrix}
  a(t) & b(t)\mathbf{1}^\mathrm{T} \\
  b(t)\mathbf{1} & \mathbf{D} \\
\end{bmatrix},
\end{equation}
where $a(t)=H_{11}$, $b(t)=H_{12}$ and $\mathbf{D}=(H_{22}-H_{23})\mathbf{U}_{n-s}+H_{23}\mathbf{1}\mathbf{1}^\mathrm{T}$. To check the existence and the expression of $\mathbf{H}^{-1}$, we first calculate the inverse of $\mathbf{D}$ by the Sherman-Morrison formula \cite{Golub_MatComp}:
\begin{equation}\label{mat_D_proof_HCRB}
\mathbf{D}^{-1}=\frac{1}{H_{22}-H_{23}}\left(\mathbf{U}_{n-s}-\frac{H_{23}}{H_{22}+H_{23}(n-s-1)}\mathbf{1}\mathbf{1}^\mathrm{T}
\right).
\end{equation}
Then the blockwise inversion formula is used for the calculation of $\mathbf{H}^{-1}$ (if it exists):
\begin{equation}
\mathbf{H}^{-1}=
\begin{bmatrix}
  f_{11}(t) & f_{12}(t)\mathbf{1}^\mathrm{T} \\
  f_{12}(t)\mathbf{1} & \mathbf{D}^{-1}
  +\frac{b^2(t)\mathbf{D}^{-1}\mathbf{1}\mathbf{1}^\mathrm{T}\mathbf{D}^{-1}}{a(t)-b^2(t)\mathbf{1}^\mathrm{T}\mathbf{D}^{-1}\mathbf{1}} \\
\end{bmatrix},
\end{equation}
where $f_{11}(t)$ and $f_{12}(t)$ are some functions of $t$. Because at last only the last $n-s$ diagonal elements will be summed up, the specific form of the two functions are of no importance. The existence of $\mathbf{H}^{-1}$ relies on whether the submatrices are valid. By employing \eqref{mat_D_proof_HCRB}, one has
\begin{equation}
\begin{aligned}
\mathbf{D}^{-1}\mathbf{1}&=\frac{1}{H_{22}+H_{23}(n-s-1)}\mathbf{1} \\
\mathbf{1}^\mathrm{T}\mathbf{D}^{-1}\mathbf{1}&=
\frac{n-s}{H_{22}+H_{23}(n-s-1)}.
\end{aligned}
\end{equation}
With these equations, it can be shown by tedious calculation that
\begin{equation}
\begin{aligned}
&\lim_{t\rightarrow0}
\frac{b^2(t)\mathbf{D}^{-1}\mathbf{1}\mathbf{1}^\mathrm{T}\mathbf{D}^{-1}}{a(t)-b^2(t)\mathbf{1}^\mathrm{T}\mathbf{D}^{-1}\mathbf{1}} \\
=~&\frac{\beta(1-2\sigma_e^2\beta)^2}{H_{22}+H_{23}(n-s-1)}
\left([H_{22}+H_{23}(n-s-1)]\right. (1+2n\sigma_e^4\beta)-\left.(n-s)\beta(1-2\sigma_e^2\beta)^2\right)^{-1},
\end{aligned}
\end{equation}
and that
\begin{equation}
\begin{aligned}
\mathbf{H}_{22}\equiv\lim_{t\rightarrow0}\left(\mathbf{D}^{-1}
  +\frac{b^2(t)\mathbf{D}^{-1}\mathbf{1}\mathbf{1}^\mathrm{T}\mathbf{D}^{-1}}{a(t)-b^2(t)\mathbf{1}^\mathrm{T}\mathbf{D}^{-1}\mathbf{1}}\right)
=\frac{(n-s){\rm e}^{-\beta}}{{\rm e}^\beta-1}
\left(1-\frac{1}{n-s+{\rm e}^\beta(1-g(\beta))^{-1}}\right),
\end{aligned}
\end{equation}
where $\mathbf{H}_{22}$ is defined as the limit of the $(2,2)$th submatrix of $\mathbf{H}$, and the function $g(\beta)$ is defined as \eqref{form_func_g}. As will be shown later, under the modest requirement that $n\geq2$, for any $\beta>0$ one has $0<g(\beta)<1$, and thus the expressions presented above are all valid. In this way we have not only checked the invertibility of $\mathbf{H}$, but also find out the expression of the inversion's limit.

To get the final form of the HCRB, we calculate $\mathbf{V}\mathbf{H}^{-1}\mathbf{V}^{\mathrm{T}}$ and sum up its last $n-s$ diagonal elements. By straightforward calculation it can be verified that this is just equal to $x_q^2\tr(\mathbf{H}_{22})$. Adding this to the lower bound of the support part, given by
\begin{equation}
\begin{aligned}
\sigma_\mathbf{x}^2\left(s-
\frac{2m\sigma_e^4\|\mathbf{x}\|^2_{\ell_2}}
{\sigma_{\mathbf{x}}^2+2m\sigma_e^4\|\mathbf{x}\|^2_{\ell_2}}\right),
\end{aligned}
\end{equation}
one will finally get the expression \eqref{HCRB_closedform}.

\blue{
The last part of the proof deals with the properties of $g(\beta)$. The proof of the two limits given by \eqref{func_g_property2} is straightforward. In order to prove \eqref{func_g_property1},
we consider separately the situations $0<\beta\leq1/(2\sigma_e^2)$ and $\beta>1/(2\sigma_e^2)$.
When $0<\beta\leq1/(2\sigma_e^2)$, one has $0\leq1-2\sigma_e^2\beta<1$, and therefore
\begin{equation*}
(1-2\sigma_e^2\beta)^2<1<1+2n\sigma_e^4\beta.
\end{equation*}
Together with $\beta<{\rm e}^\beta-1$, it leads to the inequality that
\begin{equation*}
\begin{aligned}
g(\beta)&=\frac{\beta(1-2\sigma_e^2\beta)^2}{({\rm e}^\beta-1)(1+2n\sigma_e^4\beta)}
<1,\quad\forall\beta\in\left(0,\frac{1}{2\sigma_e^2}\right].
\end{aligned}
\end{equation*}
 When $\beta>1/(2\sigma_e^2)$, i.e. $2\sigma_e^2\beta>1$, it follows that
 $0<2\sigma_e^2\beta-1<2\sigma_e^2\beta$, and thus
\begin{equation*}
\begin{aligned}
g(\beta)=\frac{\beta}{{\rm e}^\beta-1}\frac{(1-2\sigma_e^2\beta)^2}{1+2n\sigma_e^4\beta}
<\frac{\beta}{{\rm e}^\beta-1}\frac{(2\sigma_e^2\beta)^2}{1+2n\sigma_e^4\beta} <\frac{\beta}{{\rm e}^\beta-1}\frac{4\sigma_e^4\beta^2}{2n\sigma_e^4\beta}
=\frac{2}{n}\frac{\beta^2}{{\rm e}^\beta-1}.
\end{aligned}
\end{equation*}
Because $\beta^2/({\rm e}^\beta-1)<1$ for all $\beta>0$, it can be seen that $g(\beta)<1$ when $n\geq2$.}
Combining the discussions of the two situations, we have proved that $g(\beta)<1$ for all $\beta>0$. The inequality that $g(\beta)\geq0$ for all $\beta>0$ is trivial.

\section{Lower Bound of the Support Part} \label{proof_CCRB_HCRB}

In this section, we prove that the lower bound of the support part obtained by using $\mathbf{v}_i=t\mathbf{e}_i,i=1,\ldots,s$ and taking $t\rightarrow0$ is just the CCRB of maximal support. It can be easily seen that the matrix $\mathbf{V}$ in \eqref{HCRB_Vmat} is $t[\mathbf{U}_s\ \mathbf{0}]^\mathrm{T}$, and therefore the right side of \eqref{HCRB_general} is
\begin{equation}\label{cov_temp_HCRB_CCRB}
\mathbf{V}\mathbf{H}\mathbf{V}^\mathrm{T}
=\begin{bmatrix}
      t^2\mathbf{H}^\dagger & \mathbf{0} \\
      \mathbf{0} & \mathbf{0} \\
    \end{bmatrix},
\end{equation}
where the matrix $\mathbf{H}$ is given by \eqref{H_mat_HCRB_general}. Next we calculate the limit of $\mathbf{H}/t^2$. It can be seen that
for $\bs{\delta}p/pt$, one has
\begin{equation}
\left(\frac{\bs{\delta}p}{pt}\right)_i
=\frac{1}{p(\mathbf{y};\mathbf{x})}\frac{p(\mathbf{y};\mathbf{x}+t\mathbf{e}_i)-p(\mathbf{y};\mathbf{x})}{t}.
\end{equation}
Taking the limit $t\rightarrow0$, one will get
\begin{equation}
\begin{aligned}
\lim_{t\rightarrow0}\left(\frac{\bs{\delta}p}{pt}\right)_i
&=\frac{1}{p(\mathbf{y};\mathbf{x})}\lim_{t\rightarrow0}\frac{p(\mathbf{y};\mathbf{x}+t\mathbf{e}_i)-p(\mathbf{y};\mathbf{x})}{t} \\
&=\frac{1}{p(\mathbf{y};\mathbf{x})}\frac{\partial p(\mathbf{y};\mathbf{x})}{\partial \mathbf{e}_i}
=\frac{\partial\ln p(\mathbf{y};\mathbf{x})}{\partial \mathbf{e}_i},
\end{aligned}
\end{equation}
where $\partial/\partial\mathbf{e}_i$ denotes the directional derivative along $\mathbf{e}_i$ with respect to $\mathbf{x}$. \blue{
Thus the limit of $\bs{\delta}p/pt$ is
\begin{equation}
\begin{aligned}
\lim_{t\rightarrow0}\frac{\bs{\delta}p}{pt} &=
\begin{bmatrix}
  \partial\ln p/\partial\mathbf{e}_1 \\
  \vdots \\
  \partial\ln p/\partial\mathbf{e}_s \\
\end{bmatrix}
=\begin{bmatrix}
  \mathbf{U}_s & \mathbf{0} \\
\end{bmatrix}
\begin{bmatrix}
  \partial\ln p/\partial\mathbf{e}_1 \\
  \vdots \\
  \partial\ln p/\partial\mathbf{e}_s \\
  \partial\ln p/\partial\mathbf{e}_{s+1}\\
  \vdots \\
  \partial\ln p/\partial\mathbf{e}_n\\
\end{bmatrix} \\
&=
\begin{bmatrix}
  \mathbf{U}_s & \mathbf{0} \\
\end{bmatrix}
\nabla_\mathbf{x}\ln p(\mathbf{y};\mathbf{x}),
\end{aligned}
\end{equation}}
and the limit of $\mathbf{H}/t^2$ is
\begin{equation}
\begin{aligned}
\lim_{t\rightarrow0}\frac{1}{t^2}\mathbf{H}=E\left[\lim_{t\rightarrow0}\frac{\bs{\delta}p}{pt}\frac{\bs{\delta}p^\mathrm{T}}{pt}\right] =&
\begin{bmatrix}
  \mathbf{U}_s & \mathbf{0} \\
\end{bmatrix}
E\left[(\nabla_\mathbf{x}\ln p(\mathbf{y};\mathbf{x}))(\nabla_\mathbf{x}^\mathrm{T}\ln p(\mathbf{y};\mathbf{x}))\right]
\begin{bmatrix}
  \mathbf{U}_s \\
  \mathbf{0}
\end{bmatrix} \\
=&\begin{bmatrix}
  \mathbf{U}_s & \mathbf{0} \\
\end{bmatrix}\mathbf{J}\begin{bmatrix}
  \mathbf{U}_s \\
  \mathbf{0}
\end{bmatrix}.
\end{aligned}
\end{equation}
Here $\mathbf{J}$ is just the FIM, and thus the above matrix is invertible in the setting given by Section \ref{sec:Fund_problem}. Substituting it into \eqref{cov_temp_HCRB_CCRB}, and take the sum of the first $s$ diagonal elements, one will get
\begin{equation}
\begin{aligned}
\lim_{t\rightarrow0}\tr(t^2\mathbf{H}^{-1})&=\tr\left(\left(
\begin{bmatrix}
  \mathbf{U}_s & \mathbf{0} \\
\end{bmatrix}\mathbf{J}\begin{bmatrix}
  \mathbf{U}_s \\
  \mathbf{0}
\end{bmatrix}\right)^{-1}\right)
\end{aligned}.
\end{equation}
Comparing this with the proof of Theorem \ref{thm:CCRB}, one will readily accept that this bound is the CCRB for maximal support case.

\blue{\section{An Example for Section \ref{analysis_CCRB}}\label{eg_analy_CCRB}}

\blue{In this section we give an extreme example showing that utilizing the matrix perturbation appropriately might help
estimate the sparse vector $\mathbf{x}$ more accurately.}

\blue{We only consider the case where the additive noise $\mathbf{n}$ vanishes and the sensing matrix is a unit matrix. We also set the sparsity level
to be one, i.e. $s=1$, and denote the only element of the support as $k$. In this case, it can be seen that
as long as the total noise is small, the support could be correctly recovered with high probability by selecting the element with the maximum magnitude.}

\blue{A commonly used method to estimate the non-zero coefficient is given by
\begin{equation}\label{ls_estim}
\hat{x}_{\hat{k}} = y_{\hat{k}},
\end{equation}
where $\hat{k}$ denotes the estimated support element. This estimation is actually a least squares estimation. However, as the measurement
can be formulated as
\begin{equation*}
\mathbf{y}=\mathbf{x}+\mathbf{E}\mathbf{x},
\end{equation*}
it can be seen that the noise term $\mathbf{E}\mathbf{x}$ also contains some information about the sparse vector $\mathbf{x}$. The above equation
can be equivalently represented as
\begin{equation}
\mathbf{y}=\mathbf{x}+x_k\mathbf{e} = [x_ke_1,\ldots,x_k(1+e_k),\ldots,x_ke_n]^\mathrm{T},
\end{equation}
where $\mathbf{e}=[e_1,\ldots,e_n]^\mathrm{T}\sim\mathcal{N}(0,\sigma_e^2\mathbf{U}_n)$. Intuitively, as long as $n$ is sufficiently large, $e_1^2+\ldots+e_n^2$ can be arbitrarily close to $n\sigma_e^2$. Therefore, it can be seen that
\begin{equation}\label{estim_formula_last_app}
\sum_{j\neq k}y_j^2+(y_k-x_k)^2\simeq n\sigma_e^2x_k^2.
\end{equation}
One can construct an estimator from the above expression by interpreting it as an equality and substituting $\hat{k}$ for $k$ and $\hat{x}_{\hat{k}}$ for $x_k$.}

\blue{Set $n=10000$ and $\sigma_e=0.01$. In this case the estimator derived from \eqref{estim_formula_last_app} is given by
\begin{equation}\label{better_estim}
\hat{x}_{\hat{k}}=\frac{\sum_{j=1}^n y_j^2}{2y_{\hat{k}}}.
\end{equation}
We compare the two estimators \eqref{ls_estim} and \eqref{better_estim} by numerical simulations. The sparse vector is set to be
$\mathbf{x}=[1,0,\ldots,0]^\mathrm{T}$. We run the simulation 10000 times to obtain the average performance of the two estimators. The results is given in Table \ref{perform_last_app}, where the theoretical MSE of the least squares estimator is simply given by
$\sigma_e^2\|\mathbf{x}\|_{\ell_2}^2$. It can be seen that the empirical MSE of the least squares estimator is evidentally smaller than the MSE of
\eqref{ls_estim}. Therefore it can be concluded that by appropriately using the information from the noise term $\mathbf{E}\mathbf{x}$, one can
estimate the sparse vector with higher accuracy than the commonly used least squares approach.}

\begin{table}[!t]
\renewcommand{\arraystretch}{1.3}
\centering
\blue{\caption{Performance of Least Squares Estimator and Eq. \eqref{better_estim}}
\label{perform_last_app}
\begin{tabular}{c|c}
\hline
Estimator & MSE \\
\hline
Least Squares & $1\times 10^{-4}$ (Theoretical) \\
Least Squares &  $1.0161\times 10^{-4}$ (Empirical) \\
Eq. \eqref{better_estim} & $4.9823\times 10^{-5}$ (Empirical) \\
\hline
\end{tabular}}
\end{table}

\bibliographystyle{IEEEbib}
\bibliography{main_bib}
\end{document}